\documentclass[9pt,preprint,nonatbib]{sigplanconf}
\usepackage{balance}
\usepackage{float}
\usepackage{mathrsfs}
\usepackage{booktabs}
\usepackage{stmaryrd}
\usepackage{etex}
\usepackage{boxedminipage}
\usepackage[T1]{fontenc}
\usepackage{epsfig}
\usepackage{multirow}
\usepackage{url}
\usepackage[normalem]{ulem}
\usepackage{schemepgm}
\usepackage{graphicx, subcaption}
\usepackage{graphpap}
\usepackage{tabularx}
\usepackage{amssymb}
\usepackage{amsmath}
\usepackage{pst-text}
\usepackage{pst-node}
\usepackage{pst-tree}
\usepackage{bcprules}
\usepackage[boxed]{algorithm2e}
\usepackage{makecell}
\usepackage[numbers,sort&compress]{natbib}
\usepackage{enumitem}
\setlist{nolistsep}

\newcommand{\rot}[1]{\rotatebox{30}{\scalebox{0.8}{#1}}}
\newcommand{\rotmore}[1]{\rotatebox{20}{\scalebox{0.8}{#1}}}

\newcounter{lcount}

\newcommand{\cmt}[1]{}

\protect{}
\protect{}
\protect{\newlength{\BRACE}}
\newcommand{\epath}{\mathsf{ep}}

\newcommand{\acar}{\ensuremath{\mathbf{0}}}
\newcommand{\acdr}{\ensuremath{\mathbf{1}}}
\newcommand{\bcar}{\ensuremath{\bar\acar}}
\newcommand{\bcdr}{\ensuremath{\bar\acdr}}
\newcommand{\clazy}{\ensuremath{{\mathbf{2}}}}

\newcommand{\Lfun}[3]{\ensuremath{\mathcal{L}(#1,#2,#3)}}
\newcommand{\Lfunonly}{\ensuremath{\mathcal{L}}}

\newcommand{\Df}[2]{\ensuremath{\mathsf{D}_{\mathit #1}^{#2}}}

\newcommand{\Lf}[3]{\ensuremath{\Lfonly_{\mathit #1}^{#2}( {\mathit #3})}}
\newcommand{\Lfonly}{\ensuremath{\mathsf{LF}}}
\newcommand{\Lfone}[1]{\ensuremath{\Lfonly_{\mathit #1}}}

\newcommand{\Lv}{\ensuremath{\mathsf{L}}}

\newcommand{\Lanv}[2]{\ensuremath{\Lv_{{#1}}^{#2}}}

\newcommand{\eqdef}{\ensuremath{=}}

\newcommand{\var}[1]{\ensuremath{\langle #1\rangle}}

\newcommand{\mainpgm}{\ensuremath{\mathbf{main}}}
\newcommand{\main}{\ensuremath{\mathbf{main}}}

\newcommand{\figrule}{}

\newcommand{\qed}{\hfill\ensuremath{\square}}
 \newenvironment{proof}[1][Proof]{\begin{trivlist}
 \item[\hskip \labelsep {\bfseries #1}]}{\hfill\qed\end{trivlist}}

\newcommand{\len}{\ \vdash^l\ }

\newcommand{\heap}{\ensuremath{H}}       % heap
\newcommand{\pp}[2]{\ensuremath{#2}} % prog point

\def\drawplusplus#1#2#3{\hbox to 0pt{\hbox to #1{\hfill\vrule height #3 depth
      0pt width #2\hfill\vrule height #3 depth 0pt width #2\hfill
      }}\vbox to #3{\vfill\hrule height #2 depth 0pt width
      #1 \vfill}}
\newcommand{\added}[1]{#1}
\newcommand{\todelete}[1]{}%{\color{Myred}{#1}}}
\newcommand{\warning}[1]{}
\newcommand{\cred}[1]{\colorbox{lightgray}{\ensuremath{#1}}}

\newcommand{\addSubFig}[1] {
  \begin{subfigure}{\textwidth}
  \centering
    \includegraphics[width=.99\textwidth]{#1.pdf}
  \end{subfigure}
  \begin{subfigure}{\textwidth}
  \centering
    \includegraphics[width=.99\textwidth]{#1_win.pdf}
  \end{subfigure}
}

\def\myvec{\mathaccent"017E } 
\newcommand{\stk}{{S}}      

\newcommand{\bang}{\mbox{\sc bang}}
\newtheorem{theorem}{Theorem}[section]
\newtheorem{proposition}[theorem]{Proposition}
\newtheorem{definition}[theorem]{Definition}
\newtheorem{lemma}[theorem]{Lemma}
\begin{document}

\setlength{\pdfpageheight}{\paperheight}
\setlength{\pdfpagewidth}{\paperwidth}

% --- Author Metadata here ---
%\conferenceinfo{WOODSTOCK}{'97 El Paso, Texas USA}
%\CopyrightYear{2007} % Allows default copyright year (20XX) to be
%over-ridden - IF NEED BE.
%\crdata{0-12345-67-8/90/01}  % Allows default copyright data
%(0-89791-88-6/97/05) to be over-ridden - IF NEED BE. 
% --- End of Author Metadata ---

\title{Liveness-Based Garbage Collection for Lazy Languages}

% i.e., three 'name/affiliation blocks' be placed beneath the title.
%\cmt{{
%
%\numberofauthors{3} %  in this sample file, there are a *total*
% reasons) and the remaining two appear in the \additionalauthors
%section.
%
%\author{
% 1st. author
\authorinfo {Prasanna Kumar. K} 
       {IIT Bombay, Mumbai 400076, India}
       {prasannak@cse.iitb.ac.in}
% 2nd. author
\authorinfo {Amitabha Sanyal} 
       {IIT Bombay, Mumbai 400076, India}
       {as@cse.iitb.ac.in}
% 3rd. author
\authorinfo {Amey Karkare}
       {IIT Kanpur, Kanpur 208016, India}
       {karkare@cse.iitk.ac.in}

\maketitle

\begin{abstract} 
We consider  the problem of reducing  the memory required to  run lazy
first-order functional programs. Our approach is to analyze
programs  for liveness  of  heap-allocated data.   The  result of  the
analysis is  used to preserve  only live data---a subset  of reachable
data---during garbage  collection.  The result  is an increase  in the
garbage reclaimed and  a  reduction in  the peak memory requirement 
of programs.   While this  technique has already  been shown  to yield
benefits for  eager first-order  languages, the  lack of  a statically
determinable execution  order and  the presence  of closures  pose new
challenges  for lazy  languages.  These  require changes  both in  the
liveness analysis itself and in the design of the garbage collector.

%% We use the  result of the analysis to annotate  each potential garbage
%% collection  point  in   the  program  with  a   set  of  deterministic
%% finite-state automata (DFA) describing the liveness at that point.
To  show the  effectiveness of  our method,  we implemented  a copying
collector that uses  the results of the liveness  analysis to preserve
live  objects,  both evaluated  and  closures.   Our
experiments confirm  that for  programs running with  a liveness-based
garbage  collector, there  is a  significant decrease  in peak  memory
requirements.   In addition,  a  sizable reduction  in  the number  of
collections  ensures that  in spite  of using  a more  complex garbage
collector, the execution  times of programs running  with liveness and
reachability-based collectors remain comparable.
\end{abstract}




\section{Introduction}
\label{sec:intro}

Functional  programs  make  extensive  use  of  dynamically  allocated
memory.  The allocation is either explicit (i. e., using constructors)
or  implicit   (creating  closures).   Programs  in   lazy  functional
languages put additional  demands on memory, as  they require closures
to be carried from the point of creation to the point of evaluation.

While  the runtime  system  of most  functional  languages includes  a
garbage collector to efficiently  reclaim memory, empirical studies on
Scheme~\cite{karkare06effectiveness}   and  on
Haskell~\cite{rojemo96lag} programs have shown that garbage collectors
leave uncollected a large number  of memory objects that are reachable
but not live (here {\em live} means the object can potentially be used
by the program at a  later stage).  This results in unnecessary memory
retention.

In this paper,  we propose the use of liveness  analysis of heap cells
for garbage collection (GC) in a lazy first-order functional language.
The central  notion in  our analysis is  a generalization  of liveness
called {\em  demand}---the pattern of future  uses of the value  of an
expression.   The  analysis  has  two parts.   We  first  calculate  a
context-sensitive  summary   of  each   function  as  a   {\em  demand
  transformer} that transforms a {\em  symbolic} demand on its body to
demands  on its  arguments.   This  summary is  used  to step  through
function calls  during analysis.   The concrete  demand on  a function
body  is  obtained through  a  conservative  approximation similar  to
0-CFA~\cite{Shivers:1988} that  combines the demands on  all the calls
to  the function.   The result  of the  analysis is  an annotation  of
certain program points with  deterministic finite-state automata (DFA)
capturing the  liveness of  variables at  these points.   Depending on
where GC  is triggered, the  collector consults  a set of  automata to
restrict reachability during marking.  This  results in an increase in
the garbage reclaimed and consequently in fewer collections.

\added{While  the idea  of  using static  analysis  to improve  memory
  utilization  has  been  shown  to   be  effective  for  {\em  eager}
  languages~\cite{asati14lgc,       HofmannJ03,       inoue88analysis,
    lee05static}, a straightforward extension  of the technique is not
  possible  for  lazy  languages,  where  heap-allocated  objects  may
  include   closures    (runtime   representations    of   unevaluated
  expressions).   }   \todelete{Designing  a  liveness  based  garbage
  collector   for  lazy   languages  poses   significant  challenges.}
Firstly, since  data is made  live by  evaluation of closures,  and in
lazy languages  the place in  the program where this  evaluation takes
place cannot  be statically determined, laziness  complicates liveness
analysis itself.  Moreover, for liveness-based  GC to be effective, we
need to  extend it  to closures  apart from  evaluated data.   Since a
closure can escape the scope in which  it was created, it has to carry
the  liveness  information  of   its  free  variables.   As  execution
progresses  and possible  future uses  are eliminated,  we update  the
liveness  information  in  a  closure with  a  more  precise  version.
For  these  reasons,  the  garbage   collector  also  becomes
significantly more  complicated than a liveness-based  collector for an
eager language.

Experiments    with   a    single    generation   copying    collector
(Section~\ref{sec:experiments})   confirm  the   expected  performance
benefits.  Liveness-based collection results in an increase in garbage
reclaimed.  As  a consequence, there is  a reduction in the  number of
collections from 1.6X  to 23X and a decrease  in the minimum
memory requirement from  1X to 1389X.  As  an added benefit,
there is also a reduction in the overall execution time in 5 out of 12
benchmark programs.

%------------------------------------------------------------%
\subsection{Motivating example}
\label{sec:motiv}

\begin{figure}[t!]
  \includegraphics[width=\columnwidth]{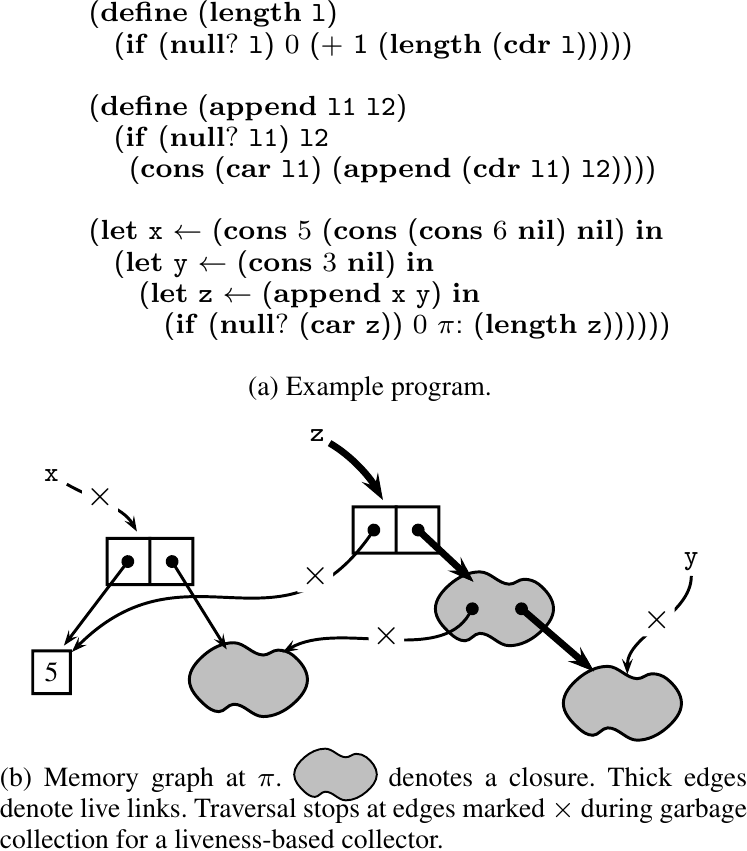}
  \caption{Example Program and its Memory Graph}\label{fig:mot-example}
\end{figure}

  Figure~\ref{fig:mot-example}   shows  an   example
program and the state of the heap at the program point $\pi$\footnote{We write
$\pi\!:\!e$ to  associate a  label $\pi$ with  the program  point just
before  the  expression  $e$. The label is not part of the syntax.},  i.e., just  before  the  evaluation  of
$(\length\ \pz)$.  The heap is represented  by a graph in which a node
either  represents  atomic  values   ($\NIL$,  integers,  etc.),  or  a
\CONS\  cell  containing  $\CAR$  and  $\CDR$  fields,  or  a  closure
(represented  by  shaded  clouds).   Edges   in  the  graph  are  {\em
  references} and represent values of variables or fields.  The figure
shows the  lists \px\  and \pz\  partially evaluated  due to  the
\SIF\ condition 
(\NULLQ~(\CAR~\pz)).
The edges  shown by thick  arrows are those  which are live  at $\pi$.
%% Only those  cell which have a future use  should be preserved
%% during garbage collection, all other
%% cells can be  reclaimed.
Thus if a GC takes  place at $\pi$ with the heap shown
in Figure~\ref{fig:mot-example}(b),  a liveness-based  collector (LGC)
will preserve only  the cells referenced by $\pz$, and  the live cells
constituting the  closure referenced by $(\CDR~\pz)$.   In contrast, a
reachability-based collector (RGC) will preserve all cells.

In  this work  we show  that  static analysis  of heap  data can  help
garbage  collectors   in  reclaiming   more  garbage.    The  specific
contributions of this paper are:
\added{
\begin{itemize}
\item We  propose an interprocedural liveness  based GC
  scheme  for a  lazy first-order functional  language and  prove its
  correctness.  To the  best of our knowledge, this is  the first work
  that uses  the results  of an  interprocedural liveness  analysis to
  garbage    collect    both     evaluated    data    and    closures.
  Thomas~\cite{Thomas19951} describes a  copying garbage collector for
  the Three  Instruction Machine (TIM)~\cite{Fairbairn1987}  that only
  preserves live closures  in a function's environment  (also called a
  frame). However, in the absence of  details, it is not clear whether
  a) the  scope of the  method is  interprocedural, and b)  it handles
  algebraic datatypes like lists (the original design of TIM did not).
  All  other  previous attempts~\cite{shaham01heap,  ran.shaham-sas03,
    shaham02estimating,  asati14lgc, karkare06effectiveness}  involved
  either imperative or eager functional languages.

\item  We  formulate a  liveness  analysis  for the  lazy  first-order
  functional language  and prove  its correctness. The  proof involves
  specifying  liveness   for  the  language  through   a  non-standard
  semantics and then proving the  analysis correct with respect to the
  specification.

\item The  analysis results  in a set  of context-free  grammars along
  with  a fixed  set  of non-context-free  productions.  The  decision
  whether to copy a cell during  GC translates to a membership problem
  for such  grammars.  Earlier research assumed  the undecidability of
  this membership question and  used an over-approximation to overcome
  it.  In this paper, we provide  a formal proof of the undecidability
  of this problem.

\item We have implemented a garbage collector that uses the result of
  liveness  analysis to  retain  live cells.   Our experiments  reveal
  interesting  space-time   trade-offs  in  the  engineering   of  the
  collector---for  example, updating  liveness information  carried in
  closures during  execution results in more  garbage being collected.
  Empirical results  show the effectiveness of  liveness-based GC.
\end{itemize}
}
\todelete{
\begin{itemize}
\item  We  formulate  a  liveness  analysis for  a  lazy  first-order
  functional language and  prove its correctness. The 
  proof involves specifying liveness through a non-standard semantics  and then proving the analysis 
  correct with respect to the specification.   
\item The analysis results  in a set  of context-free  grammars along
  with  a fixed  set  of non-context-free  productions.  The  decision
  whether   to  copy  a  cell   during  GC  translates  to a  membership problem  for such  grammars.  In  this
  paper, we  have shown the  undecidability of this problem and 
  circumvent the  problem by over-approximating the  grammars
  with DFAs.
\item We  have implemented a garbage  collector that uses the  DFAs to
  retain live  cells.  Our  experiments reveal  interesting space-time
  trade-offs  in  the  engineering  of  the  collector---for  example,
  updating liveness  information carried in closures  during execution
  results  in   more  garbage  being  collected.    Empirical  results
  show the effectiveness of liveness-based GC.
\end{itemize}

As far  as we know, this  is the first work  on liveness-based GC      for      lazy       languages.       All      previous
attempts~\cite{shaham01heap,   ran.shaham-sas03,   shaham02estimating,
  asati14lgc, karkare06effectiveness} have  involved either imperative
or eager functional languages.
}
\subsection{Organization of the paper}

Section~\ref{sec:defs}  introduces  the   syntax  of  the  programming
language considered  and gives a small-step  operational semantics for
it.  The liveness  analysis for this language and  its soundness proof
is          presented          in          Section~\ref{sec:liveness}.
Section~\ref{sec:computing} describes  the formulation of  liveness as
grammars. We also give a proof  of undecidability of such grammars and
show     how     they     can     be     approximated     by     DFAs.
Section~\ref{sec:GC-scheme}   discusses   details   of   the   garbage
collector,  in  particular  the  use  of  liveness  DFAs  for  GC.     We     report    our    experimental     results    in
Section~\ref{sec:experiments}    along    with   some    observations.
Section~\ref{sec:relatedwork}  discusses  previous   work  related  to
GC  and   liveness  and  Section~\ref{sec:conclusion}
discusses possible extensions and concludes the paper.
\section{The target language---syntax and semantics}
\label{sec:defs}
Figure~\ref{fig:lang-syntax} describes the syntax  of our language. It
is a first order language with lazy semantics. Programs are restricted
to        be        in        Administrative        Normal        Form
(ANF)~\cite{chakravarty03perspective}, where  all actual  parameters to
functions  are  variables.  While  this  restriction  does not  affect
expressibility,  this form  has  the benefit  of  making explicit  the
creation of closures through the $\LET$ construct.   We
  assume that  $\LET$s  in  our  language  are  non-recursive;  in  the
  expression $\LET\,\, x  \leftarrow s\,\, \IN\,\, e$,  $x$ should not
  occur in  $s$. The restriction of  \LET\ to a single  definition is
for ease  of exposition---generalization to multiple  definitions does
not add conceptual difficulties.  We further restrict each variable in
a program  to be  distinct, so that  no scope  shadowing occurs; this
simplifies reasoning.

We denote the body of a function ${\mathit  f}$   as $e_{\mathit f}$.
 We assume that each program has a distinguished function
\mainpgm,    defined as\linebreak    $(\DEFINE\      ({\tt
  \mainpgm})\  e_\mainpgm)$,  and  the program begins execution
with the call to \mainpgm.  

\begin{figure}[t]\footnotesize
\renewcommand{\arraystretch}{0.9}
\begin{eqnarray*}
   p \in \mathit{Prog} & \!\!\!::=\!\!\! & d_1 \ldots d_n \,\,\,\, e_\mainpgm
   \hspace{5em} \,\,\,\,\,\,\,\,\,\; \mbox{\em --- program}\\
    \mathit{df} \in Fdef & \!\!\!::=\!\!\! & (\DEFINE\,\, (f\,\, x_1 \,\, \ldots
\,\,x_n)\,\,
    e)
    \hspace{0.2em} \ \ \ \ \ \ \ \ \  \mbox{\em --- function def} \\
e \in \mathit{Expr} & \!\!\!::=\!\!\! &
\left\{\begin{array}{@{}ll@{\hspace{1em}}l}
       (\SIF\,\, x\,\, e_1\,\, e_2) && \!\!\!\mbox{\em --- conditional} \\
       (\LET\,\, x \leftarrow s\,\, \IN\,\, e) &&\!\!\! \mbox{\em --- let
binding} \\
       (\SRETURN\,\, x) && \!\!\!\mbox{\em --- return from function}
    \end{array}\right. \\
s \in \mathit{App} & \!\!\!::=\!\!\!  &
\left\{\begin{array}{@{}l@{\hspace{1.2em}}l}
       k & \mbox{\em --- constant (numeric or $\NIL$)}\\
       (\CONS\,\, x_1\,\, x_2) & \mbox{\em --- constructor} \\
       (\CAR\,\, x) &  \mbox{\em --- selects $1^{st}$ part of \CONS} \\
       (\CDR\,\, x) &  \mbox{\em --- selects $2^{nd}$ part of \CONS} \\
       (\NULLQ\,\, x) & \mbox{\em --- returns 0 if x is not $\NIL$} \\
       (\PRIM\,\, x_1\,\, x_2) & \mbox{\em generic arithmetic} \\
%%       (\ID\,\, x) & \mbox{\em ---  identity function (for inlining)} \\
       (f\,\, x_1\,\,\ldots\,\, x_n) & \mbox{\em --- function application}
    \end{array}\right.
\end{eqnarray*}
%\vspace*{-6pt}
  \caption{The syntax of our language}\label{fig:lang-syntax}
%\figrule
\normalsize
%\vspace*{-6pt}
\end{figure}

\begin{figure*}[t!]
%\comment{

\begin{center}\footnotesize
\renewcommand{\arraystretch}{1.2}
\begin{tabular}{|c|c|c|}
\hline
Premise & Transition & Rule name \\
\hline
\hline 
          & $\rho,\, (\rho', \ell, e)\!:\!S,\, \heap,\, \kappa
  \rightsquigarrow \rho',\, S,\, \heap[\ell := \kappa],\, e$    &  {\sc const}
\\
\hline
          & {$\rho, \,(\rho', \ell, e)\!:\!S,\, \heap,\, (\CONS~x~y)
\rightsquigarrow
$  $\rho',\, S,\, \heap[\ell := (\rho(x), \rho(y))],\, e$}     &  {\sc cons} \\
\hline

$\heap(\rho(x)) \mbox{ is } (v, d)$ & $\rho,\, (\rho', \ell, e)\!:\!S,\, \heap,\,
(\CAR~x)  \rightsquigarrow \rho',\, S,\, \heap[\ell := v],\, e$      &
{\sc car-select} \\
\hline

$\heap(\rho(x)) \mbox{ is } (\langle s, \rho'\rangle, d)$ & $\rho,\, S,\,  \heap,\,
(\CAR~x)  \rightsquigarrow \rho', \,(\rho, addr(\langle s, \rho'\rangle), (\CAR~x))\!:\!S,\, \heap,\, s$      &
{\sc car-1-clo} \\
\hline

$\heap(\rho(x)) \mbox{ is } \langle s, \rho'\rangle$ & $\rho,\, S,\, \heap,\, (\CAR~x)
\rightsquigarrow
\rho',\, (\rho, \rho(x), (\CAR~x))\!:\!S, \,\heap,\, s$      &
{\sc car-clo}
\\
\hline

$\heap(\rho(x)), \heap(\rho(y)) \in \mathbb{N}$
 & {$\rho,\, (\rho', \ell, e)\!:\!S,\, \heap,\, (+~x~y)  \rightsquigarrow$
$\rho', \,S,\, \heap[\ell := \heap(\rho'(x)) + \heap(\rho'(y))],\, e$}      &
{\sc prim-add} \\
\hline

$\heap(\rho(x)) \mbox{ is } \langle s, \rho'\rangle$ & $\rho,\, S,\, \heap,\, (+~x~y)
\rightsquigarrow
\rho',\, (\rho, \rho(x), (+~x~y))\!:\!S,\, \heap,\, s$      &
{\sc prim-1-clo} \\
\hline
$\heap(\rho(y)) \mbox{ is } \langle s, \rho'\rangle $ & $\rho,\, S,\, \heap,\, (+~x~y)
\rightsquigarrow
\rho', (\rho,\, \rho(y),\, (+~x~y))\!:\!S,\, \heap,\, s$      &
{\sc prim-2-clo} \\
\hline
{$\mathit{f}~\mbox{defined as}$
$~(\DEFINE~(f~\myvec{y})~e_{\mathit{f}})$}  & $\rho,\, S,\, \heap,\,
(f~\myvec{x})  \rightsquigarrow
[\myvec{y} \mapsto \rho(\myvec{x})],\, S,\, \heap,\, e_{\mathit{f}}$      &
{\sc funcall} \\
\hline
$\ell$ is a new location& {$\rho,\, S,\, \heap,\, (\LET~x\leftarrow s~\IN~e)
  \rightsquigarrow$
$\rho\oplus[x \mapsto \ell],\, S,\, \heap[\ell := \langle s,
    \lfloor\rho\rfloor_{FV(s)} %% \oplus [x \mapsto  \ell] -- NO LAZY LETS
    \rangle],\, e$} &
{\sc let} \\
\hline
$\heap(\rho(x)) \ne 0$ & $\rho,\, S,\, \heap,\, (\SIF~x~e_1~e_2)   \rightsquigarrow
\rho,\, S,\, \heap,\,  e_1$ & {\sc if-true} \\
\hline
$\heap(\rho(x)) = 0$ & $\rho,\, S,\, \heap,\, (\SIF~x~e_1~e_2)   \rightsquigarrow
\rho,\, S, \,\heap, \, e_2$ & {\sc if-false} \\
\hline
$\heap(\rho(x)) = \langle s, \rho' \rangle $ & {$\rho,\, S,\, \heap,\,
  (\SIF~x~e_1~e_2)   \rightsquigarrow
\rho',\, (\rho, \rho(x), (\SIF~x~e_1~e_2))\!:\!S,\, \heap, \, s$}
&
{\sc if-clo} \\
\hline
{$\heap(\rho(x))~\mbox{is in WHNF with value}~v$}& $\rho,\, (\rho', \ell,
e)\!:\!S,\, \heap,\,
(\SRETURN~x)  \rightsquigarrow \rho',\, S,\, \heap[\ell := v],\, e$ &
{\sc return-whnf}\\
\hline
$\heap(\rho(x)) = \langle s, \rho' \rangle $ & {$\rho,\, S,\, \heap,\, (\SRETURN~x)
  \rightsquigarrow$
$\rho',\, (\rho, \rho(x), (\SRETURN~x))\!:\!S,\, \heap,\,  s$} &
{\sc return-clo} \\
\hline
\end{tabular}
\caption{A small-step semantics for the language. \label{fig:lang-semantics}}
\end{center}
%}\comment
%\vspace*{-16pt}
\end{figure*}

\subsection{Semantics}\label{sec:semantics}
%\comment{
We now give  a small-step semantics for our  language.
We first specify the domains used by the semantics:
\[
\renewcommand{\arraystretch}{1}
\begin{array}{@{}r@{\ }l@{\ \ }c@{\ \ }l@{\hspace{0.5em}}l}
\heap: & \mathit{Heap} & =&\mathit{Loc} \rightarrow (Data + \{empty\}) & \mbox{-- Heap}\\
d: & \mathit{Data} &=&\mathit{Val} + \mathit{Clo} & \mbox{-- Values \& Closures} \\
v:   & \mathit{Val} &=& \mathbb{N} + \{\NIL\} + \mathit{Data \times Data}& \mbox{-- Values}\\
c:   & \mathit{Clo} &=& \mathit{(App \times Env)}& \mbox{-- Closures}\\
\rho: & \mathit{Env} &=&\mathit{Var} \rightarrow \mathit{Loc} &
\mbox{-- Environment} \\
\end{array}
\]

Here $\mathit{Loc}$  is a countable set  of locations in the  heap.  A
non-empty location either contains a {\em closure}, or a value in Weak
Head Normal Form  (WHNF)\cite{Jones87}. In our case, a  WHNF value can
be a  number, the  empty list  $\NIL$ or a  \CONS\ cell  with possibly
unevaluated   constituents.   A   closure  is   a  pair   $\langle  s,
\rho\rangle$ in  which $s$ is  an unevaluated application,  and $\rho$
maps free  variables of $s$  to their respective locations.  Since all
data objects are boxed, we model  an environment as a mapping from the
set of  variables of  the program $\mathit{Var}$  to locations  in the
heap.

The  semantics  of   expressions  (and  applications\footnote{In  most
 contexts, we  shall use the  term 'expression' and the  notation $e$
to stand   for both expressions and applications.}) 
are given by transitions of the form $\rho, \stk, \heap, e \rightsquigarrow
\rho',  \stk', \heap',  e'$.  Here  $\stk$ is  a stack  of continuation
frames.  Each  continuation frame  is a triple $(\rho,  \ell, e)$,
signifying that the  location $\ell$ has to be updated  with the value
of the  currently evaluating  expression and $e$  is to  be evaluated
next  in the  environment $\rho$.  
The  initial  state  of  the  transition  system  is:
\[([\;]_\rho,\,     (\rho_\mathit{init},     \,    \ell_{\ans}     ,\,
(\print~\ans)):[\;]_{S},   [\;]_{\heap},\,   (\mainpgm))\]  in   which
$[\;]_\rho$, $[\;]_{\mbox  \footnotesize {H}}$ and $[\;]_{S}$  are are
the empty environment, heap and  stack respectively. The initial stack
consists  of  a  single  continuation   frame  in  which  \ans\  is  a
distinguished variable that will eventually  be updated with the value
of  (\mainpgm),  and  $\rho_\mathit{init}$  maps \ans\  to  a  location
$\ell_{\ans}$.  In addition, \print\ is a function modeling a printing
mechanism---a   standard   run-time   support  assumption   for   lazy
languages~\cite{Jones87}---that prints  the value of  (\mainpgm).  The
operator~$:$ pushes elements on top of the stack.

The   notation  $[\myvec{x}   \mapsto  \myvec{\ell}]$   represents  an
environment  that  maps  variables  $x_i$ to  locations  $\ell_i$  and
$\heap[\ell := d]$  indicates updation of  \heap\ at $\ell$
with  $d$.   $\rho \oplus  \rho'$  represents  the environment  $\rho$
shadowed  by  $\rho'$  and  $\lfloor \rho  \rfloor_X$  represents  the
environment  restricted  to  the  variables in  $X$.  Finally  $FV(s)$
represents the  free variables  in the  application $s$  and $addr(c)$
gives the address of the closure $c$ in the heap.

The small-step semantics  is shown in Figure~\ref{fig:lang-semantics}.
Unlike   an   eager  language,   evaluation   of   a  let   expression
$(\LET~x\leftarrow  s~\IN~e)$ does  not  result in  the evaluation  of
$s$. Instead,  as the {\sc let}  rule shows, a closure  is created and
bound to  $x$. The  program points which  trigger evaluation  of these
closures   are  an $\SIF$   condition ({\sc  IF-clo})   and
 a \SRETURN\ ({\sc  return-clo}).  We  call such  points \emph{evaluation
  points   $(\epath)$}  and   label  them   with  $\psi$   instead  of
$\pi$. As an  example of closure  evaluation, we  explain the
three rules for  $(\CAR~x)$.  If $x$ is a closure,  it is evaluated to
WHNF, say $(d_1, d_2)$.  This is  given by the rule {\sc car-clo}.  If
$d_1$ is  not in WHNF,  it is  also evaluated ({\sc  car-1-clo}).  The
address  to be  updated  with the  evaluated value  is
recorded  in  a  continuation  frame.    This  is  required  for the
evaluation  to be lazy, else $d_1$ may  be evaluated more than once
due to  sharing~\cite{Jones87}.  Only after this  is the actual selection  done ({\sc
  car-select}).

\section{Liveness}
\label{sec:liveness}

A variable is {\em live} if there  is a possibility of its value being
used in  future computations and  dead if  it is definitely  not used.
Heap-allocated data  needs a richer model  than classical liveness---a
model which talks about liveness of references.  

%% The model use the notion of  {\em access paths}, i.e.,
%% a  prefix-closed set of strings  over $\{\acar,\acdr\}^\ast$, in
%% which 
%% $\acar$,  $\acdr$ represent  access  using $\CAR$  and $\CDR$  fields.

An {\em access path}
is a  prefix-closed set of strings  over $\{\acar,\acdr\}^\ast$, where
$\acar$,  $\acdr$ represent  access  using $\CAR$  and $\CDR$  fields respectively.
Given an initial location $\ell$ (usually a reference corresponding to
a variable)  and a  heap \heap, semantically  an access  path $\alpha$
represents     a      reference,     denoted     $\heap\llbracket\ell,
\alpha\rrbracket$,  in the  heap  that is  obtained  by starting  with
$\ell$ and chasing the \CAR\ or  \CDR\ fields in the heap as specified
by  the  access   path.   $\heap\llbracket\ell,  \alpha\rrbracket$  is
defined only if  the path followed in the  heap is \emph{closure-free}
(does not cross closures), else it is undefined.

Access paths are used to  represent liveness. As an example, a list
$x$ with 
liveness  $\{\epsilon,  \acar,  \acdr,  \acdr\acar,
\acdr\acdr,  \acdr\acdr\acar\}$ means that future computations    only
refer up  to the  second and  third members of  $x$.  A  {\em liveness
  environment} is a mapping from  variables to access paths, but often
expressed as  a set,  for example  by writing  $\{x.\epsilon, x.\acdr,
x.\acdr\acdr,        y.\epsilon\}$        instead        of        $[x
  \mapsto\{\epsilon,\acdr,\acdr\acdr\},          y\mapsto\{\epsilon\},
  z\mapsto\{\}]$.    In  this   notation,  $y   \mapsto  \{\epsilon\}$
represents access using $y$ itself and $z \mapsto \{\}$ indicates 
$z$ is dead.  In lazy  languages, liveness environments are associated
with regions of programs instead of program points.

A notion  that generalizes liveness  is {\em demand}.   While liveness
gives the patterns of future uses of a variable, demand represents the
future uses of the value of an expression.  The demand on an expression
$e$   is   also   a   set    of   access   paths---the   subset   of
$\{\acar,\acdr\}^\ast$ which the  context of $e$ may  explore of $e$'s
result.   To see  the need  for demands,  consider the expression
$(\LET~x\leftarrow  (\CDR~y)~\IN~   (\SRETURN~x))$.   Assume   that  the
context of this expression places  the demand $\lbrace \epsilon, \acar
\rbrace$. Since the  value of the expression is the  value of $x$, the
demand translates to the liveness  $[x \mapsto \lbrace \epsilon, \acar
  \rbrace]$.   Due  to  the  \LET\   definition  which  binds  $x$  to
$(\CDR~y)$, the liveness of $x$  now becomes the demand on $(\CDR~y)$.
This, in  turn, generates  the liveness $\lbrace  y.\epsilon, y.\acdr,
y.\acdr\acar  \rbrace$.   These are  the $y$-rooted  accesses
required to  explore $\lbrace  \epsilon, \acar  \rbrace$ paths  of the
result of $(\CDR~y)$.
%% As an  analogy with classical  (strong) liveness analysis,  $y$ and
%% $z$ are live at the entry $\pi: x:=y+z$, if and only if $x$ is live
%% at  exit  of $\pi$.   In  our  terminology, the  liveness  $\lbrace
%% \epsilon\rbrace $ of $x$ at the  exit from $\pi$ becomes the demand
%% on  $y+z$,  and  this,  in turn  generates  the  liveness  $\lbrace
%% y.\epsilon, z.\epsilon \rbrace$ at the entry of $\pi$.
  
We use $\sigma$  to range over demands, $\alpha$ to  range over access
paths and $\Lv$  to range over liveness environments.  The liveness of
an  individual variable  $y$ in  \Lv\  is $\Lv(y)$,  but more  commonly
written as  $\Lv_y$.  The notation $\sigma_1\sigma_2$  denotes the set
$\lbrace  \alpha_1\alpha_2 \mid  \alpha_1 \in  \sigma_1, \alpha_2  \in
\sigma_2\rbrace$.  Often we shall abuse  notation to juxtapose an edge
label and  a set  of access  paths; $\acar\sigma$  is a  shorthand for
$\lbrace\acar\rbrace\sigma$.

%==============================================================
\renewcommand{\pp}[2]{\ensuremath{#1\!\!:\!#2}} % prog point
%==============================================================

\begin{figure}[t]
  \input{example-to-illustrate-liveness}  
  \caption{Example illustrating liveness of closures}\label{fig:lv-closure}
\end{figure}
\begin{figure}[t!]
  \renewcommand{\arraystretch}{1.2}
  \small
\hspace*{-6pt}
$\begin{array}{@{}r@{\ }c@{\ }l@{}}
\mathit{ref\/}(\kappa,\sigma,\Lfonly) &=& \emptyset \mbox{, for $\kappa$ a constant, including \NIL} \\

\mathit{ref\/}(\pp{\pi}{(\CONS~x~y)},\sigma,\Lfonly)
&=& \{x_{\pi}.\alpha \mid \acar\alpha \in \sigma\} \cup
\{y_{\pi}.\alpha 
\mid \acdr\alpha \in \sigma\} \\

\mathit{ref\/}(\pp{\pi}{(\CAR~x)},\sigma,\Lfonly)
&=&\!\left\{\!\!\!\!\begin{array}{l@{\ }l}
  \{x_{\pi}.\epsilon\} \cup \{x_{\pi}.\acar\alpha \mid \alpha \in
  \sigma\}, & \mbox{if}~\sigma \ne \emptyset\\
  \emptyset  & \mbox{otherwise}
\end{array}\right. \\

\mathit{ref\/}(\pp{\pi}{(\CDR~x)},\sigma,\Lfonly)
&=&\!\left\{\!\!\!\!\begin{array}{l@{\ }l}
\{x_{\pi}.\epsilon\} \cup \{x_{\pi}.\acdr\alpha \mid \alpha \in
\sigma\}, & \mbox{if}~\sigma \ne \emptyset\\
\emptyset  & \mbox{otherwise}
\end{array}\right. \\

\mathit{ref\/}(\pp{\pi}{(\PRIM~x~y)},\sigma,\Lfonly)
&=&\!\left\{\!\!\!\!\begin{array}{l@{\ }l}
\{x_{\pi}.\epsilon, y_{\pi}.\epsilon\},  & \mbox{if}~\sigma \ne
\emptyset\\
\emptyset  & \mbox{otherwise}
\end{array}\right.\\

\mathit{ref\/}(\pp{\pi}{(\NULLQ~x)},\sigma,\Lfonly)
&=&\!\left\{\!\!\!\!\begin{array}{l@{\ }l}
\{x_{\pi}.\epsilon\},  & \mbox{if}~\sigma \ne \emptyset\\
\emptyset  & \mbox{otherwise}
\end{array}\right. \\

\mathit{ref\/}(\pp{\pi}{(f~\myvec{x})},\sigma,\Lfonly)
%          &=& \bigcup_{i=1}^n y_i.\Lf{f}{i}{\sigma}
&=&  \begin{array}{@{}l}  % to discourage \displaystyle
  \bigcup_{i=1}^n x_{i_{\pi}}.\Lf{f}{i}{\sigma}
\end{array}
%          &=& \bigcup \{y_i.\Lf{f}{i}{\sigma} \mid i=1,\ldots, n\}
\\

\mathcal{L}((\SRETURN~\pp{\psi\ }{x}),\sigma,\Lfonly) &=& 
\Lv_{\emptyset} \cup \{x.\sigma\}\\

\mathcal{L}((\SIF~\pp{\psi\ }{x~e_1~e_2}),\sigma,\Lfonly) &=&
\!\left\{\!\!\!\!\begin{array}{l@{\ }l}
  \mathcal{L}(e_1,\sigma,\Lfonly) \cup
  \mathcal{L}(e_2,\sigma,\Lfonly) \cup
  \{x.\epsilon\},  & \\ 
  \qquad\qquad\qquad\qquad\qquad\qquad\mbox{if}~\sigma \ne \emptyset & \\
  \emptyset  \qquad\qquad\qquad\qquad\qquad\quad\; \mbox{otherwise} &\\
\end{array}\right. \\

\mathcal{L}((\LET~x \leftarrow~\pi:s~\IN~e),\sigma,\Lfonly) &=&
\mathit{ref\/}(s,\sigma',\Lfonly) \cup \Lv \cup \{x.\sigma'\} 
\cup \Lv''
\\
&&\hspace*{-.65cm}~\mbox{where}~ \mathcal{\Lv} = \mathcal{L}(e,\sigma,\Lfonly),\\
&&~\sigma' = \bigcup_{\pi}  \Lv(x_{\pi})\\
&&~\Lv' = \mathit{ref}(s, \sigma', \Lfonly)\\
&&~\Lv'' = \underset{\py \in FV(s)}{{\bigcup}} \lbrack\py \mapsto \Lv(\py) \cup \Lv'(\py_{\pi})\rbrack
\\ & \\
\end{array}$%\vspace*{-3mm}
\begin{center}
\begin{minipage}{0.95\columnwidth}

  \infrule[live-define]
          {\mathcal{L}(e_f,\sigma,\Lfonly) =
            \bigcup_{i=1}^n z_i.\Lf{f}{i}{\sigma}
            \mbox{ for each $f$ and $\sigma$}
          }
          { \mathit{df_1} \ldots \mathit{df_k} \len \Lfonly
            \\ \makebox[20mm]{\rule{15mm}{0pt}where
              $(\DEFINE\ (f\ z_1\ \ldots\ z_n)\ \ e_f)$ 
              is a member of $\mathit{df_1}
              \ldots \mathit{df_k}$}}
\end{minipage}
\end{center}
\caption{Liveness equations and judgement rule}\label{fig:live-judge}
%\vspace*{-12pt}
\end{figure}

\subsection{Liveness analysis for lazy languages}
\label{sec:liveness-analysis}

Consider the  program in Figure~\ref{fig:lv-closure}.   As mentioned
earlier, a lazy evaluation of the $\LET$ expression at $\pi_1$ creates
a closure  for $(\length~\px)$  instead of  evaluating it.   Since the
closure may escape the scope in which it is created, it carries a copy
of $\px$ within itself.  We treat the  copy of $\px$ in the closure as
being separate from the $\px$ introduced  by the $\LET$, and call it a
\emph{closure  variable}.  For  liveness calculations,  such variables
are distinguished  from variables  introduced by  \LET s  and function
arguments  (called \emph{stack  variables},  since they  reside in  the
activation  stack).   We notationally distinguish  a closure  variable
from  its corresponding  stack variable  by subscripting  it with  the
label of the program   point  where   the   closure  was   created\footnote{Multiple
  occurrences  of the  same  variable in  an  application are  further
  distinguished by their positions in the application.}.

  Since a  closure is evaluated  only at evaluation points,  a closure
  variable is attributed  with the same liveness in  the entire region
  of  the  program from  the  point  of  creation  of the  closure  to
  reachable evaluation points.  This is  also true of stack variables,
  because, as we shall see, stack variables derive their liveness from
  closure variables. Thus, there are two major differences between our
  formulation of  liveness of  lazy languages  with liveness  of eager
  languages~\cite{asati14lgc}:   (i)  the   introduction  of   closure
  variables in the  liveness calculations, and (ii)  a single liveness
  value for each  variable that is applicable from  its creation point
  to evaluation points.

Closure  variables  get  their  liveness values  through  a  chain  of
dependences beginning  at a  variable at an  evaluation point.   As an
example,  in  Figure~\ref{fig:lv-closure},   a  dependence  chain  for
$\px_{\mathit{\pi}_1}$  begins   with  the   variable  $\pz$   at  the
evaluation  point $\psi_3$.  The  variable \pz\  returned at  $\psi_3$
depends on  \py\ through  the expression  $(+~\py~1)$.  $\py$  in turn
depends  on   the  closure  variable   $\px_{\mathit{\pi}_1}$  through
$(\length~\px_{\mathit{\pi}_1})$.  We denote this chain of dependences
as   $\lbrack  \psi_3\!\!\!:\!\!\!\pz   \leftarrow  (+~\py~1),\,   \py
\leftarrow (\length~\px_{\mathit{\pi}_1})\rbrack$.  Indeed, the chains
of  closures  in  the  heap  are a  runtime  representation  of  these
dependences.   Since  $\pz$  is  evaluated  at  $\psi_3$  due  to  the
expression $\SRETURN~\pz$,  the demand made by  the calling context(s)
of  $f$ places  a demand  on  $\pz$ which  will impart  a liveness  to
$\px_{\pi_1}$.  Other dependence chains which result in a liveness for
$\px_{\pi_1}$    are    $\lbrack    \psi_1\!\!\!:\!\!\py    \leftarrow
(\length~\px_{\pi_1})\rbrack$   and  $\lbrack   \psi_2\!\!\!:\!\!\!\pw
\leftarrow          (/~\pu~\py),\,\,           \py          \leftarrow
(\length~\px_{\pi_1})\rbrack$.   The  liveness analysis  described  in
this section declares  the liveness of $\px_{\pi_1}$ to be  a union of
the liveness arising  from these dependence chains.  To be  safe, a GC
during evaluation  of $\py$ at  $\psi_1$ has  to use this  liveness to
copy the heap starting from  $\px_{\pi_1}$.  However, notice that if a
GC  takes place  while evaluating  $\pz$  at $\psi_3$,  it can  safely
consider only the liveness arising from the dependence chain\linebreak
$\lbrack \psi_3\!\!\!:\!\!\!\pz \leftarrow (+~\py~1),\, \py \leftarrow
(\length~\px_{\pi_1})\rbrack$.    The    garbage   collection   scheme
described in Section~\ref{sec:GC-scheme} uses a generalization of this
observation  to  dynamically  select   an  evaluation  point  specific
liveness in order to collect more garbage.

Figure~\ref{fig:live-judge}  describes  our  analysis  which  has  two
parts. The  function $\mathit{ref}$,  takes an  application $s$  and a
demand $\sigma$ and returns the incremental liveness generated for the
free variables of $s$ due to  the application.  This will be consulted
during GC  while exploring the heap  starting from the
closure  variables.  The function $\mathcal{L}$  uses  $\mathit{ref}$ to  propagate
liveness across expressions.

In  a   lazy  language,   an  expression   is  not   evaluated  unless
required. Therefore  the null  demand ($\emptyset$) does  not generate
liveness in any of the rules defining $\mathit{ref}$ or $\mathcal{L}$.
A non-null  demand of  $\sigma$ on (\CDR~$x$),  is transformed  to the
liveness $\{x.\epsilon,  x.\acdr\sigma\}$.  In an opposite  sense, the
demand  of  $\acdr\sigma$ on  (\CONS~$y$~$z$)  is  transformed to  the
demand  $\sigma$  on  $z$.   Since \CONS\  does  not  dereference  its
arguments, there  is no $\epsilon$ demand  on $y$ and $z$.   The rules
for (\PRIM~x~y) and (\NULLQ~x) are  similar. Constants do not generate
any liveness.

In case of a  function call, we  use the third  parameter $\Lfonly$
that  represents  the  summaries  of all  functions  in  the  program.
$\Lfonly_{\mathit  f}$  (the summary for a  specific
function $f$) expresses  how the demand $\sigma$ on a  call to $f$ is
transformed into  the liveness of  its parameters at the  beginning of
the  call.  $\Lfonly$  is determined  by the  judgement $\mathit{Prog}
\len  \Lfonly$ using  inference rule  ({\sc live-define}).   This rule
describes  the  fixed-point property  to  be  satisfied by  $\Lfonly$,
namely, the  demand transformation  assumed for  each function  in the
program should  be the  same as  the demand  transformation calculated
from      its      body.       As       we      shall      see      in
Section~\ref{sec:grammar-formulation},  we  convert  the rule  into  a
grammar and  the language generated  by this grammar is  the least
solution satisfying  the rule. We  prefer the least solution  since it
ensures the safe collection of the greatest amount of garbage.

We next  describe the function $\mathcal{L}$  that propagates liveness
across  expressions.   Consider  the $\mathcal{L}$-rules  for  {\LET},
{\SIF}, and {\SRETURN}.  Since the value of $(\SRETURN~x)$ is the
value  of  $x$,  a  demand $\sigma$  on  $(\SRETURN~x)$  gives  a
liveness   of  $\{x.\sigma\}$.    The  liveness   of  the   expression
$(\SIF~x~e_1~e_2)$  is  a union  of  the  liveness of  $e_1$  and
$e_2$. In  addition, since  the condition $x$  is also  evaluated, the
liveness $\{x.\epsilon\}$ is created and  added to the union.  
To  understand  the  liveness rule  for  $\LET~x  \leftarrow~s~\IN~e$,
observe that the value  of $\LET$ is the value of  its body $e$.  Thus
the  liveness environment  $\Lv$ of  $e$ is  calculated for  the given
demand $\sigma$. Since the stack variable $x$ is copied to each of the
closure variable  $x_{\pi}$, the liveness of  $x$ is the union  of the
liveness of the  closure variables.  This liveness,  say $\sigma'$, is
also  the  demand  on  $s$, thus the  liveness  environment
$\mathit{ref}(s,   \sigma',   \Lfonly)$   is  added   to   $\Lv   \cup
\{x.\sigma\}$. Finally, the stack variables corresponding to the free
variables of $s$ are updated and added 
  to  give the  overall liveness environment  for $(\LET~x
\leftarrow~s~\IN~e)$.

\added{  As noted  earlier, $x.\alpha  \in \Lv$  specifies the
  liveness of the reference $\heap\llbracket\rho(x), \alpha\rrbracket$
  only   if   $\alpha$  corresponds   to   a   closure-free  path   in
  \heap\ starting  from $\rho(x)$.  If  this path is intercepted  by a
  closure,  say  $(\CAR~{y_{\pi}})$, then  the  liveness  of the  path
  starting from $\mathit{y_{\pi}}$ is given by $\Lv_\mathit{y_{\pi}}$.
  As we shall see in  Section~\ref{sec:GC-scheme}, the liveness of the
  closure variable  $y_{\pi}$ is recorded  along with the  closure for
  $s$ so that the GC can refer to it during garbage collection.  }
%} %comment 

\subsection{Soundness of analysis}  

We shall now present a proof of soundness of the analysis presented in
Section~\ref{sec:liveness-analysis}. \added{It is easy to see that the analysis
correctly identifies the liveness of stack variables. A stack variable
is live
between its introduction through a \LET\  and its last use to
create a closure variable. This is correctly captured by the
\LET\ rule in Figure~\ref{fig:live-judge}.}
Proving soundness  for closure  variables is more complex. Here  are the
ideas behind the proof. 
\begin{enumerate}
\item      We     augment      the      standard     semantics      in
  Figure~\ref{fig:lang-semantics} to  model a GC before  the execution
  of  each \LET\  expression.  \added{Note  that, unlike  eager
    languages,   memory  is   allocated  only   during  execution   of
    \LET\ expressions}. During GC, we track each reference in the root
  set and heap that is declared  dead by our analysis.  Any attempt to
  dereference  such  references  results   in  the  transition  system
  entering a special state denoted \bang.  We call the semantics after
  augmentation, \emph{minefield semantics}.

%% During GC, we bind  a special value
%%   $\bot$ to each  reference in the root  set and heap  that is not
%%   live by  our analysis.  This  models the de-allocation  of locations
%%   without any  live references.  Any  attempt to dereference  a $\bot$
%%   during execution results
%% %\added{that  has   been   declared  not   live  by   our
%% %    analysis} %%with  a null-liveness during a  transition will result
%%   in the transition system entering  a special state denoted \bang.  We
%%   call the semantics after augmentation, \emph{minefield semantics}.
\item \label{inline1} Assuming that a program enters the \bang\ state,
  we  construct,  through  inline expansion,  a  program
  without function  calls which has  the same minefield  behavior. The
  final step shows that no program  without function calls can enter the
  \bang\ state. As a consequence  no program (with or without function
  calls) can enter the \bang\ state.
\end{enumerate}

\begin{figure*}[t!]
\begin{center}\footnotesize
\renewcommand{\arraystretch}{1.2}
\begin{tabular}{|@{}c@{}|@{}c@{}|@{}c@{}|}
\hline
Premise & Transition & Rule name \\ 
\hline
\hline 
\makecell[t]{ $\rho(x) \mbox{ is } \bot$} & $\rho, S,
  \heap, (\CAR~x), \cred{\sigma} \rightsquigarrow \bang$ &
{\sc car-bang} 
\\
\hline
\makecell[t]{$\heap(\rho(x)) \mbox{ is } (\langle s, \rho'\rangle, d)$} & \makecell[t]{$\rho,\, S,\,  \heap,\,
(\CAR~x), \cred{\sigma} \rightsquigarrow$ $ \rho', \,(\rho, addr(\langle
s, \rho'\rangle), (\CAR~x),\cred{\sigma} )\!:\!S, \heap,\, s, \, \cred{\sigma}$ }     &
{\sc car-1-clo} 
\\
\hline
\makecell[t]{ $\heap(\rho(x)) \mbox{ is } \langle s, \rho'\rangle$} & 
\makecell[t]{ $\rho, S, \heap, (\CAR~x), \cred{\sigma}
  \rightsquigarrow$ 
  $\rho', (\rho, \rho(x), (\CAR~x), \cred{\sigma})\!:\!S, \heap, s,\cred{\renewcommand{\arraystretch}{1}\left\{
    \begin{array}{@{}lr@{}}
      {\emptyset}&\mbox{if } \sigma = \emptyset \\       
      {\{\epsilon\} \cup \acar\sigma }&\mbox{otherwise}
    \end{array}\right.
  }$   }      
&
{\sc car-clo}\\
\hline
\makecell[t]{$\cred{GC(\rho_1, S_1, \heap_1, (\LET~x\leftarrow
    s~\IN~e), \sigma) = (\rho, S, \heap)}$,\\$\ell$ is a new location}& \makecell[t]{$\rho, S, \heap, (\LET~x\leftarrow
  s~\IN~e), \cred{\sigma}$  $ \rightsquigarrow \rho\oplus[x
    \mapsto \ell], S, \heap[\ell := \langle s, \lfloor\rho\rfloor_{FV(s)}, \sigma_x\rangle], e, \cred{\sigma}$ \\
    where $\sigma_x\ =\  \lfloor\mathcal{L}(e,\sigma,\Lfonly)\rfloor_{\{x\}}$} &
{\sc let} \\ 
\hline
\end{tabular}
\end{center}
\caption{Minefield  semantics for  \CAR\   and \LET.
  \warning{Add a note mentioning that we handle the issue of requiring
    a      specific     demand      by      introducing     a      new
    symbol}\label{fig:minefield-semantics-for-some}}
%\vspace*{-2mm}
\end{figure*}
\noindent To set up the minefield semantics, we follow these steps:
\begin{enumerate}

\item  We  extend  the  abstract  machine  state  $\rho,S,\heap,e$  to
  $\rho,S,\heap,e,\sigma$.   We call  such a  state a  \emph{minefield
  state}.  Here $\sigma$  is the ``dynamic'' demand  on the expression
  $e$, that arises from the actual sequence of function calls that led
  to  the  evaluation of  $e$.\footnote{The  static  liveness that  is
    consulted during actual GC  is computed from an over-approximation
    of this demand.  Thus the soundness  result on the modeled GC will
    also apply for the actual GC.} The demand for the initial state is
  $\sigma_{\mathit  {all}}$,  and each  $\rightsquigarrow$  transition
  transforms   the  demand   according  to   the  liveness   rules  of
  Section~\ref{sec:liveness-analysis}.      The     information     in
  continuation frames  on the stack  $S$ are also  similarly augmented
  with their demands.  Thus a stacked entry now takes the form $(\rho,
  \ell, e, \sigma)$.

\item The closure created by a let expression is now a 3-tuple
  $\langle s,  \rho, \sigma \rangle$, where $\sigma$ represents the
  demand on the closure. 

\item  $GC(\rho,  S,  \heap,  e, \sigma)$ models  a  liveness-based  garbage
  collection that returns $(\rho', S', \heap')$. The changes in
  $\rho, S$ and  $\heap$ are due
  to non-live references being replaced by $\bot$.  This simulates the
  act of  garbage collecting  the cells
  pointed to by these references  during an actual garbage collection.
  To do so, $GC(\ldots)$ needs to consider the following environments:  (1) the  environment in  the
  current state, (2) the  environment in each of the stacked  continuations and (3)
  the environment  in each of  the closures in  the heap.
  
 \begin{enumerate} 
\added{

 \item \label{env}  For each of  these environments, $GC(\ldots)$  calculates a
   liveness  environment $\Lv$  using  the  corresponding  $e$  (or  $s$)  and
   $\sigma$.
 \item 
   For each
    location $\ell$,  $GC(\ldots)$  sets $\heap(\ell)$  to $\bot$ iff for each
    environment $\rho$ above, for each $x \in domain(\rho)$, and each forward
    access
    path $\alpha$, it is \emph{not} the case that $x.\alpha \in \Lv$ and $\heap\llbracket\rho(x), \alpha
    \rrbracket = \ell$. 
%% $x.\alpha$  is  live  and
%%     corresponds  to   $\ell$ in \heap.  More  formally,  there   are  no  $x\in
%%     domain(\rho)$ and  $\alpha$, such that  $\heap\llbracket\rho(x), \alpha
%%     \rrbracket = \ell$  and $x.\alpha \in \Lv$, where 
}

  \end{enumerate}
\end{enumerate}
Figure ~\ref{fig:minefield-semantics-for-some} shows some of the
minefield rules. As mentioned earlier, the transition for a \LET\ 
%%  $([\;]_\rho,([\;]_\rho,
%% \ans, (\print~\ans)):[\;]_{S}  , [\;]_{\heap}, (\mainpgm)$,
is preceded  by $GC(\ldots)$.  The details of the transition for the {\sc
  car-clo}  rule   is  also shown.  If    an
earlier  call to  $GC(\ldots)$  results  in $\rho(x)$  being
bound  to  $\bot$,  then  the  $\rightsquigarrow$  step
enters the  \bang\ state ({\sc  car-bang}). Otherwise
the transition is similar to the earlier {\sc car-clo}
rule. The remaining rules for minefield semantics are given in Appendix~\ref{sec:minefield}.

Consider a trace  of a minefield execution of a  program $p$, possibly
ending in  a \bang\ state. We  can construct a call-tree  based on the
trace in  which each node represents  a function that was  called (but
did not necessarily  return because of a \bang).  Assume  that each of
the nodes of  the tree is also annotated with  the program point where
the corresponding call  was invoked.  This tree can be  used to inline
function calls in a hierarchical fashion.  The details of the inlining
can be found in~\cite{asati14lgc}.
For  a
call-less program,  the initial state of  the minefield
semantics  is  assumed  to  be  $([\;]_\rho,(\rho_\mathit{init},
\ans,      (\print~\ans)):[\;]_{S}     ,      [\;]_{\heap},
e_{\mainpgm})$.

\subsection{Soundness result}
We first need an auxiliary  result about minefield semantics. Consider
a trace of  a minefield execution.  For every  minefield state $(\rho,
S,   \heap,  e(s),   \sigma)$   that   appears  on   the   LHS  of   a
$\rightsquigarrow$ step, the demand $\sigma$ on the expression $e$ (or
application $s$) is  non-null.  This can be proved by  an induction on
the number  of steps leading  to the  minefield state.  The  base step
holds   because  the   demand   $\sigma_{\mathit{all}}$   on  $(\mainpgm)$   is
non-null. For the inductive step we  observe that for each step of the
minefield semantics, if the demand $\sigma$  on the LHS of a minefield
step  is non-null,  the  demand  on the  RHS  is  a transformation  of
$\sigma$ (for example $(\clazy\cup\acar)\sigma$) which is also non-null.

Note that our proofs will be for a single round of minefield execution
i.e.,  the evaluation of  $(\mainpgm)$  to   its  WHNF  driven  by   the  printing  mechanism~(Section
\ref{sec:semantics}).  With  minor variations,  the proof will  also be
applicable for subsequent evaluations initiated by $\print$.
%% We show  next that the  minefield execution of a  call-less program
%% cannot go \bang.

\begin{lemma}
\label{lemma:call-less-cannot-go-bang}
Consider the minefield execution of  a program without function calls.
Such a program cannot enter the \bang\ state.
\end{lemma}
\begin{proof}
Consider a  state  $(\rho,  S,  \heap,  e, \sigma)$  in  the
minefield execution of a program.  We  show by induction on the number $n$
of $\rightsquigarrow$ steps  leading to  this state  that the next  transition cannot
enter a  \bang\ state.   When $n$  is 0,  the state  is $([\;]_\rho,\,
(\rho_\mathit{init},   \,  \ell_{\ans}   ,\,  (\print~\ans)):[\;]_{S},
[\;]_{\heap},\, e_{\mainpgm})$.   Since the call  to $GC(\ldots)$ in  this state
does  nothing,  we  just  have to  show  that  the  $\rightsquigarrow$
transition cannot  enter a  \bang\ state.  Since  our programs  are in
ANF, $e_{\mainpgm}$ can only be a $\LET$ expression.  A {\sc let} step
does not involve dereferencing, and thus cannot result in a \bang.

For the inductive step, we shall show that none of the minefield steps
that involves dereferencing  results in a \bang.  These  are the steps
which   have   a   $\heap(\rho(...))$   in   the   premise.    Now   a
$\rightsquigarrow$ step can go \bang\ because it dereferences a $\bot$
inserted by an earlier $GC(\ldots)$.   However the demand $\sigma'$  on basis of
which the $GC(\ldots)$ would have inserted a  $\bot$ would have
included the  current demand $\sigma$.     Thus it is
enough to  show that the $\rightsquigarrow$ step would not lead to a
\bang, even if $GC(\ldots)$ had been done with the current demand
$\sigma$.

We consider the rules for \CAR\  only.  
The rest  of the rules involve similar reasoning. For the {\sc
  car-clo}  rule in  the state  $\rho, S,  \heap, (\CAR~x),
\sigma$, we  know that $\sigma$ is  non-null. Therefore
the  liveness  of  $x$  includes  $\epsilon$,  and  the
dereferencing $\heap(\rho(x))$ will go without \bang.

For the  {\sc car-1-clo}  rule, observe that  there are
two dereferences.  First $x$  is dereferenced to  get a
cons  cell  and then  the  head  of  the cons  cell  is
dereferenced  to  obtain  a  closure.   If  the  demand
$\sigma$ on  $(\CAR~x)$ is non-null, then  the liveness
of    $x$   will    include    both   $\epsilon$    and
$\acar\epsilon$,  and  a  GC with  this  liveness  will
neither bind $x$ to a  $\bot$, nor insert $\bot$ at the
first   component  of   the   cons   cell.  Thus   both
dereferences    can take place without 
$\rightsquigarrow$ entering the   \bang\ state.
\end{proof} 

Now we are ready to prove the main soundness result.

\begin{theorem}
The  minefield  execution of  no  program  can enter  a
\bang\ state.
\end{theorem}
 
\begin{proof}
Assume to  the contrary that  a program $P$  enters the
\bang\  state.  We can  transform  $P$  to a  call-less
program $P'$ such that  the minefield executions of $P$
and $P'$  are identical  except for change  of variable
names.             However,                 by
Lemma~\ref{lemma:call-less-cannot-go-bang} we know that
$P'$,   a   call-less   program,   cannot   enter   the
\bang\  state.  Therefore   $P$  also  cannot  enter  the
\bang\ state.
\end{proof}
 
%% Section~\ref{sec:computing}   shows   how  the   demand   transformers
%% $\Lfonly$  for  a  program  (representing  a  fully  context-sensitive
%% analysis) can be safely approximated  by a {\em procedure summary} for
%% each function.   The summary is  in the  form of a  demand transformer
%% that maps a demand on a call to the function to demands on each of its
%% arguments.

\section{Towards a  computable form of liveness}\label{sec:computing}
%\section{Generating and storing liveness information}\label{sec:computing}
The analysis in Section~\ref{sec:liveness} is fully context-sensitive, 
describing the  liveness sets  in a function body in
terms of a symbolic demand $\sigma$  and \Lfonly. However, we have yet to
describe (i) how to obtain demand transformers \Lfonly\ from the rule
{\sc live-define} , and (ii) how to compute the concrete demand $\sigma$ on
each function. To do so, we first need to modify the liveness rules
to a slightly different form.

%\subsection{Modifying liveness rules}
\paragraph{Symbolic representation of operations:}
The   $\mathit{ref}$   rule   for   \CONS,   shown   in
Figure~\ref{fig:live-judge}, requires us  to remove the
leading  \acar\ and  \acdr\  from the  access paths  in
$\sigma$.  Similarly, the rules  for \CAR, \CDR, \PRIM,
\NULLQ, and \SIF\ require  us to return $\emptyset$, if
$\sigma$      itself       is      $\emptyset$      and
$\lbrace\epsilon\rbrace$  otherwise.  To  realize these
rules  $\sigma$   needs  to  be  known.   This  creates
difficulties  since  we  want to  solve  the  equations
arising from liveness symbolically.

The solution is to  also treat the operations mentioned
above  symbolically.  We  introduce three  new symbols:
\bcar, \bcdr,  \clazy.  These symbols are  defined as a
relation  $\hookrightarrow$  between   sets  of  access
paths:
\begin{align*}
  &\bcar\sigma \hookrightarrow \sigma' \mbox{ where } \sigma' = \{\alpha \mid \acar\alpha \in \sigma\}\\
  &\bcdr\sigma \hookrightarrow \sigma' \mbox{ where } \sigma' = \{\alpha \mid \acdr\alpha \in \sigma\}
\end{align*}
Thus \bcar\ selects those entries in $\sigma$ that have leading \acar, and removes the leading \acar\ from them.
The symbol \clazy\ reduces the set of strings following it to a set containing only $\epsilon$. It filters out, however, the empty set of strings.
\begin{align*}
  \clazy\sigma \hookrightarrow & \left\{ 
  \begin{array}{ll}
    \emptyset&\mbox{if}~\sigma = \emptyset\\
    \{\epsilon\} & \mbox{otherwise}
  \end{array}\right.
\end{align*}
We can  now rewrite the \CONS\  and the \CAR\  rules of $\mathit{ref}$
as:
\begin{align*}
&\mathit{ref\/}((\CONS~x~y),\sigma,\Lfonly)
= x.\bcar\sigma \cup y.\bcdr\sigma  \label{eqn:mod-cons},~
\mbox{and} \\
&\mathit{ref\/}((\CAR~x),\sigma,\Lfonly)
          =   x.\clazy\sigma \cup x.0\sigma
\end{align*}
and the \Lfunonly\ rule
for \SIF\ as:
\begin{align*}
\mathcal{L}((\SIF~x~e_1~e_2),\sigma,\Lfonly) =
                    &\mathcal{L}(e_1,\sigma,\Lfonly)~\cup
        \mathcal{L}(e_2,\sigma,\Lfonly)~\cup
          \{x.\clazy \sigma\}
\end{align*}
The rules for  \CDR, \PRIM\ and \NULLQ\ are also
modified similarly.

When   there  are   multiple   occurrences  of   \bcar,
\bcdr\  and \clazy,  $\hookrightarrow$ is  applied from
right  to left.   The reflexive  transitive closure  of
$\hookrightarrow$      will      be     denoted      as
$\stackrel{*}{\hookrightarrow}$.      The     following
proposition  relates  the  original  and  the  modified
liveness rules. 
\begin{proposition}
Assume  that  a  liveness   computation  based  on  the
original  set  of  rules  gives  the  liveness  of  the
variable  $x$   as $\sigma$
(symbolically,  $\Lanv{x}{}=  \sigma$).  Further,  let
$\Lanv{x}{}= \sigma'$ when the modified rules are used
instead      of      \Lfunonly.      Then      $\sigma'
\stackrel{*}{\hookrightarrow} \sigma$.
\end{proposition}

An  explanation of  why the  proposition holds  for the
modified  \CONS\ rule  is  given in  \cite{asati14lgc}.
The proposition also holds for other modified rules for
similar reasons.

\begin{figure}[t!]
\renewcommand{\arraystretch}{1}
\begin{minipage}{.40\textwidth}
  \small
  \renewcommand{\arraystretch}{1}{
    \begin{uprogram}
      \UFL(\DEFINE\ (\length~\xl)
      \UNL{1}  $\pi_1\!\!:\, $(\LET\ \px\ $\leftarrow $\ (\NULLQ~\xl) \IN
      \UNL{2} \hspace*{.05cm} $\pi_2\!\!:\,$(\SIF\
      $\psi_1\!\!:\,$ \px
      \UNL{3} \hspace*{.27cm} $\pi_3\!\!:\,
      $(\LET\ \pv\ $\leftarrow 0$\ \IN \hspace*{1.4mm}$\pi_4\!\!:\,(\SRETURN~\psi_2:\pv)$
      \UNL{3} \hspace*{.29cm}    $\pi_5\!\!:\, $(\LET~\pu\
      $\leftarrow$  (\CDR~\xl)  \IN
      \UNL{4} \hspace*{.34cm}   $\pi_6\!\!:\, $(\LET~\py\
      $\leftarrow$  (\length~\pu)  \IN
      \UNL{5} \hspace*{.34cm} $\pi_7\!\!:\,
      $(\LET~\pz\ $\leftarrow$ (+~1~\py)\ \IN \hspace*{1.4mm} $\pi_8\!\!:\,
      (\SRETURN~\psi_3:\pz)$)))))))
  \end{uprogram}}
  \renewcommand{\arraystretch}{1}{
    \begin{uprogram}
      \UFL $(\DEFINE\ (\main)$
      \UNL{1} \!\!$\pi_9\!\!:\, (\LET\  \pa\  \leftarrow$
      (
      %\scalebox{0.8}{\psframebox[framearc=.5,
      %fillcolor=lightgray,fillstyle=solid,framesep=2pt]{%
      %    \begin{tabular}{@{}c@{}}
            \cred{\mbox{a BIG closure}}
      %\end{tabular}}}
      ) \IN  
      \UNL{2} \!\!$\pi_{10}\!\!:\, (\LET\  \pb\  \leftarrow ($+$\ \pa\ \acdr)$
      \IN
      \UNL{3}   \hspace*{.05cm}$\pi_{11}\!\!:\,      $ (\LET\ \pc\
      $\leftarrow  (\CONS\ \pb\ \NIL)$ \IN
      \UNL{4}   \hspace*{.15cm}    $\pi_{12}\!\!:\,
      $(\LET\ \pw\  $\leftarrow  (\length\ \pc)$ \IN  \hspace*{1.4mm}  $\pi_{13}\!\!:\,
      (\SRETURN~\psi_4:\pw)))))$
  \end{uprogram}}
  %%}
\end{minipage}
\caption{An example program}\label{fig:mot-example2-a}
%\vspace*{-5mm}
\end{figure}
%\vspace*{-3mm}

%\subsection{Generating equations for $\Lfonly_{\mathit f}$}
\paragraph{Computing function summaries $\Lfonly_{\mathit f}$:}
Given a  function $\mathit{f}$, we now  describe how to
generate  equations   for  the   demand  transformation
\Lfonly$_\mathit{f}$.        The       program       in
Figure~\ref{fig:mot-example2-a}   serves  as   a  running
example.  Starting with a  symbolic demand $\sigma$, we
determine  \Lfun{e_{\mathit{f}}}{\sigma}{\Lfonly}.   In
particular, we  consider $\Lanv{x_i}{}$,  the liveness
of the $i^{\mbox{\footnotesize th}}$ parameter $x_i$.   By   the   rule   {\sc
  live-define},   this   should    be   the   same   as
$\Lf{f}{i}{\sigma}$. Applying this to \length, we have:
$$
 \Lf{\length}{1}{\sigma} = \Lanv{\xl}{} = \clazy\sigma \cup \acdr\Lf{\length}{1}{\clazy\sigma}
  \cup \clazy\Lf{\length}{1}{\clazy\sigma}
$$

In general, the equations for \Lfonly\ are recursive as in the case of
\Lfonly$_\mathit{f}$. A closed  form solution for $\Lfone{\mathit{f}}$
can be derived by observing that each of the liveness rules modifies a
demand  only by  prefixing it  with symbols  in the  alphabet $\lbrace
\acar, \acdr,\bcar,  \bcdr, \clazy  \rbrace$. Therefore we  can assume
that $\Lf{f}{i}{\sigma}$ has the closed form:
\begin{eqnarray}
\label{eq:LF:DI}
  \Lf{f}{i}{\sigma} = \Df{f}{i}\sigma
\end{eqnarray}
where \Df{f}{i}  are sets of strings  over the alphabet
mentioned above.  Substituting the  guessed form in the
equation   describing    \Lfonly$_{\mathit   f}$,   and
factoring  out   $\sigma$,  we  get  an   equation  for
\Df{f}{i}  that   is  independent  of   $\sigma$.   Any
solution   for   \Df{f}{i}   yields  a   solution   for
\Lfonly$_{\mathit  f}$.   Applied to  \Lfonly$_{\mathit
  \length}$, we get:
  \begin{eqnarray*}
&&  \Lf{\length }{1}{\sigma} = \Df{\length}{1}\sigma,~\mbox{and}\\
&&   \Df{\length}{1} = \clazy \cup \acdr\Df{\length}{1}\clazy
       \cup \clazy\Df{\length}{1}{\clazy}
  \end{eqnarray*}

Note that this equation can also be viewed as a CFG with \{\acdr,
\clazy\} as terminal symbols and \Df{\length}{1} as the sole
non-terminal.

%\subsection{Generating liveness equations \Lv\  for function bodies}
\paragraph{Handling user-defined functions:}
\label{sec:bodylivenessbodies}
To  avoid analyzing  the body  of a  function for  each
call, we calculated the  liveness for the arguments and
the variables in a function  with respect to a symbolic
demand  $\sigma$.   To  get   the  actual  liveness  we
calculate an  over-approximation of the  actual demands
made by  all the  calls and  calculate the  liveness at
each  GC  point  inside  the  function  based  on  this
approximation.  The 0-CFA-style {\em summary demand} is
calculated by  taking a union  of the demands  at every
call site of a function.

%% Consider a function  $g$ containing a call to  $f$ at a
%% site    $\pi$,   say    $\pi\!\!:\!(\LET~x   \leftarrow
%% (f\,y_1\,\ldots\,y_n)~\IN~\pi_i:e)$.  Let the demand on
%% $g$  be  $\sigma_g$  and,  based on  this  demand,  the
%% liveness of  $x$ at  $\pi_i$ be $\Lanv{i}{x}$.   By the
%% $\LET$ rule of Figure~\ref{fig:live-judge}, the call at
%% $\pi$   contributes   $\Lanv{i}{x}$   to   the   demand
%% $\sigma_f$. Let  us denote  the contribution of  a call
%% site $\pi$ in  a function $g$ to the  overall demand on
%% the function  $f$ as  $\deltacall{f}{\pi}{g}$. Assuming
%% that there are $k$ call  sites to function $f$, $\pi^1$
%% (in function $g^1$) \ldots $\pi^k$ (in function $g^k$),
%% the     over-approximation     of     $\sigma_f$     is
%% $\deltacall{f}{\pi^1}{g^1}     \cup     \cdots     \cup
%% \deltacall{f}{\pi^k}{g^k}$.  The distinguished function
%% \mainpgm\ is  a special case.   We assume it  is called
%% through    a    printing    mechanism    with    demand
%% $\sigma_\mainpgm   =   \{\acar,\acdr\}^\ast$   (denoted
%% $\sigma_{\!all}$) if \mainpgm\  returns a structure and
%% $\epsilon$ if it returns a base value.

For  the  running  example (Figure~\ref{fig:mot-example2-a}),  $\length$  has  calls  from  $\main$  at
$\pi_{12}$ and a recursive call at $\pi_6$. The demands on these calls
are $\epsilon$ and $\clazy\sigma_{\mathit{\length}}$. Thus:
\begin{eqnarray*}
\sigma_{\length}    &=&
 \{\epsilon\}  ~\cup~{\clazy\sigma_{\mathit{\length}}}
\end{eqnarray*}
As examples, the closure variables $\xl_{\mbox{$\pi$}_1}$ and $\xl_{\mbox{$\pi$}_5}$, and the
stack variables $\pl$ and $\pa$ have the following liveness in terms
of  $\sigma_{\length}$:
\begin{align*}
  \Lv_{\xl_{\mbox{$\pi$}_1}} &= \clazy\sigma_{\mathit{\length}} \\
%\end{align*}
%\begin{align*}
  \Lv_{\xl_{\mbox{$\pi$}_5}} &=  (\acdr\Df{\length}{1}{\clazy} \cup
  \clazy\Df{\length}{1}{\clazy}) \sigma_{\mathit{\length}} \\
  \Lanv{\pl}{} &= (\clazy \cup \acdr\Df{\length}{1}\clazy
  \cup \clazy\Df{\length}{1}{\clazy})\sigma_{\mathit{\length}} \\
  \Lanv{\pa}{} &= \clazy \bcar \Df{\length}{1}\{\epsilon\}
\end{align*}

In  summary, the  equations  generated during the
analysis are:
\begin{enumerate}
\item   For  each   function  $\mathit{f}$,   equations
  defining \Df{f}{i} for use by \Lfonly$_{\mathit f}$.
\item For each function $\mathit{f}$, an equation defining the summary
  demand $\sigma_{\mathit f}$ on $e_f$.
\item   For  each   function  $\mathit{f}$   (including
  $\mainpgm$) an  equation  defining
  liveness for each garbage in  $e_{\mathit f}$.
\end{enumerate}

%\subsection{Solving liveness equations---the grammar interpretation}\label{sec:grammar-formulation}      
\paragraph{From liveness sets to context-free grammars:}\label{sec:grammar-formulation}      
The  equations  above can  now  be  re-interpreted as  a  context-free
grammar  (CFG) on  the  alphabet $\lbrace\acar,  \acdr, \bcar,  \bcdr,
\clazy\rbrace$.  Let \var{$X$} denote  the non-terminal for a variable
$X$ occurring on the LHS of the equations generated from the analysis.
We can  think of  the resulting productions  as being  associated with
several grammars, one for  each non-terminal \var{\Lanv{x}{}} regarded
as a start symbol.  As  an example, the grammar for \var{\Lanv{\xl}{}}
and \var{\Lanv{\pa}{}} comprise the following productions:
\begin{eqnarray*}
  \var{\Lanv{\xl}{}}  &\rightarrow& 
  \clazy\var{\sigma_{\length}} \mid \acdr \var{\Df {\length}{1}}{\clazy\var{\sigma_{\length}}}  \\ && \mid
  \clazy\var{\Df{\length}{1}}{\clazy\var{\sigma_{\length}}} \\
  \var{\Df{\length}{1}} &\rightarrow& \clazy \mid
  \acdr\var{\Df{\length}{1}}\clazy
       \mid \clazy\var{\Df{\length}{1}}{\clazy}\\
\langle {\sigma_{\length}} \rangle
&\rightarrow&
\epsilon  \mid \clazy\langle{\sigma_{\length}}\rangle \\
\var{\Lanv{\pa}{}} &\rightarrow& \clazy \bcar \var{\Df{\length}{1}}
\end{eqnarray*}
Other  equations  can   be  converted  similarly.   The
language   generated   by  \var{\Lanv{x}{}},   denoted
$\mathscr{L}(\var{\Lanv{x}{}})$,   is    the   desired
solution  of  \Lanv{x}{}.    
However, recall that the decision problem that we are interested in
during GC is: 
%%\begin{quote} 
{\em
Let $x.\alpha$ be a forward access path (consisting only
of edges  \acar\ and \acdr\  but not \bcar,  \bcdr\ or
\clazy).       Let      $\mathscr{L}(\var{\Lanv{x}{}})
\stackrel{*}{\hookrightarrow}  \sigma$, where  $\sigma$
consists  of  forward  access  paths  only.  Then  does
$\alpha \in \sigma$?
}
%%\end{quote}
We model  this problem as one  of deciding the
membership of  a CFG augmented  with a
fixed set of unrestricted productions.

\begin{definition}\label{def:specialgrammar}
Consider the  grammar $(N,T,  p_1\cup p_2,S)$  in which
$N$ is  a set  of non-terminals,  $T =  \{\acar, \acdr,
\bcar,  \bcdr,  \clazy,  \$\}$,   $p_1$  is  a  set  of
context-free    productions     that    contains    the
distinguished  production   $S  \rightarrow  \alpha\$$,
$\alpha$ is a  string of grammar symbols  that does not
contain $S$, and $p_2$ is the fixed set of unrestricted
productions    $\bcar\acar    \rightarrow    \epsilon$,
$\bcdr\acdr    \rightarrow   \epsilon$,    $\clazy\acar
\rightarrow \clazy$,  $\clazy\acdr \rightarrow \clazy$,
and $\clazy\$ \rightarrow \$ $.
\end{definition}

%% From Sections \ref{sec:liveness}  and \ref{sec:computing}, it is clear
%% that the results of liveness analysis of any program can be modeled by
%% the kind  of grammar described above. The  following proposition shows
%% that the converse also holds.
%% \begin{proposition}
%% Given    a     grammar    $G$    of    the     form    described    in
%% Definition~\ref{def:specialgrammar},  it is  possible  to construct  a
%% program $p$  with program points  $\pi_i$ and variables $x$  such that
%% the liveness  analysis of $p$  is same as  $G$ except for a  change in
%% non-terminal names which now have the form \var{\Lanv{i}{x}}.
%% \end{proposition}
We first show that membership problem of the class of grammars in
Definition~\ref{def:specialgrammar} is undecidable. However note that
the grammars may not be necessarily generated from liveness analysis.
\newcommand{\state}{\ensuremath{\mathsf{S}}}
\newcommand{\nont}[2]{\ensuremath{\mathsf{S}_{#1}^{#2}}}  
\begin{lemma}\label{lemma:grammar-undecidable}
Given a grammar    $G$    of   the    kind    described    in
Definition~\ref{def:specialgrammar}  and a forward access path $\alpha$
consisting  of symbols \acar\  and \acdr\  only, the membership problem
$\alpha\$ \in \mathscr{L}(G)$ is undecidable.
\end{lemma} 

\begin{proof}
Given a Turing machine and  an input $w\in (1+0)^*$, we
construct a grammar $G$ such that the machine will halt
on    $w$    if    and   only    if    $\$    \in
\mathscr{L}(G)$. The grammar includes  the fixed set of
unrestricted               productions               in
Definition~\ref{def:specialgrammar}.

We shall represent a Turing machine (TM) configuration 
as $w_l(\state,c)w_r$, where $w_l$ is the string to the
left of the head, $w_r$ is the string to the right, $c$
is the symbol under the head and \state\ is the current
state of  the TM.   For each combination  of state
and symbol  $(\state,c)$, the grammar will  contain the
non-terminal \nont{}{c}. We shall synchronize each move
of the TM to  a derivation  step using  a context
free production, followed, if possible, by a derivation
step using either  $\bcar\acar \rightarrow \epsilon$ or
$\bcdr\acdr   \rightarrow    \epsilon$.    After   each
synchronization,  we  shall   establish  the  following
invariant  relation between  the TM  configuration
and the sentential form:

\begin{quote}
  If    the   configuration    of   the    TM   is
  $w_l(\state,c)w_r$, then the  sentential form will be
  $\overline{w}_l\nont{}{c}\,w_r\$ $,                where
  $\overline{w}_l$ is  the same as $w_l$  but with each
  symbol $d$ in $w_l$ replaced by $\overline{d}$.
\end{quote}

Assume that  the TM starts  in a  state $S_\mathit{init}$ with  a tape
$cw$ and  the head positioned on  the symbol $c$. Then  the sentential
form    corresponding     to    the    initial     configuration    is
$\nont{\mathit{init}}{c}w\$$ (we  can assume that there  is a production
$\nont{}{} \rightarrow  \nont{\mathit{init}}{c}w\$$, where  \nont{}{} is
the start symbol of the  grammar). Further correspondences between the
TM moves and the grammar productions are as follows:

\begin{enumerate}
\item For each transition  $(S_i, c) \rightarrow (S_j,c',L)$, there are
  two  productions $\nont{i}{c}  \rightarrow  \acar \nont{j}{\acar}c'$
  and $\nont{i}{c} \rightarrow \acdr \nont{j}{\acdr}c'$.
\item For each  transition $(S_i, c) \rightarrow (S_j,c',R)$, there
  are two productions
  $\nont{i}{c} \rightarrow \overline{c'} \nont{j}{\acar}\bcar$ and $\nont{i}{c}
  \rightarrow \overline{c'} \nont{j}{\acdr}\bcdr$. 
\end{enumerate}
The idea  behind the  productions is explained  with an
example:  Assume that  the current  sentential form  is
$\bcar\bcdr\nont{i}{\acar}\acar\acar\$$. Also  assume that  the
TM  has  a   transition  $(S_i,\acar)  \rightarrow
(S_j,\acdr, L)$.  Since the next corresponding  step in
the  derivation  has  to  be  done  without  any  prior
knowledge of whether the symbol to the left of the tape
is a \acar\  or a \acdr, two  productions are provided,
and  the  invariant  will  be maintained  only  if  the
production         $\nont{i}{\acar}         \rightarrow
\acdr\nont{j}{\acdr}\acdr$ is chosen  for the next step
in  the  derivation.    This  gives  the  configuration
$\bcar\bcdr\acdr\nont{j}{\acdr}\acdr\acar\acar\$$.
Simplification   with    the   production   $\bcdr\acdr
\rightarrow               \epsilon$              yields
$\bcar\nont{j}{\acdr}\acdr\acar\acar\$$,    which   exactly
corresponds  to   the  changed  configuration   of  the
TM.  Notice  carefully that a wrong  choice breaks
the invariant  and it cannot be  recovered subsequently
by any choice of productions.

After the  TM has  halted, there  are further  ``cleanup'' derivations
that derive  $\epsilon$ only if  the invariant has been  maintained so
far.    For   every   symbol   $c$,  we   introduce   a   non-terminal
$\nont{\mathit{final}}{c}$  where   $\nont{\mathit{final}}{}$  is  the
final state  of the  TM. We add  productions $\nont{\mathit{final}}{c}
\rightarrow         \acar        \nont{\mathit{final}}{c}$         and
$\nont{\mathit{final}}{c} \rightarrow  \acdr \nont{\mathit{final}}{c}$
for cleaning up the \bcar\ and \bcdr\  symbols on the left of the head
and                $\nont{\mathit{final}}{c}               \rightarrow
\nont{\mathit{final}}{c}\bcar$      and      $\nont{\mathit{final}}{c}
\rightarrow \nont{\mathit{final}}{c}\bcdr$ for  cleaning up \acar\ and
\acdr\ on the right of the tape head.  This completes the
reduction.\!
\end{proof}

We now show that the proof can be replayed for the class of grammars
generated from liveness analysis of programs. 

\begin{lemma}\label{lemma:grammar-from-analysis-undecidable}
Given   a    grammar    $G$    of   the    kind    described    in
Definition~\ref{def:specialgrammar}  that  is  generated  by  liveness
analysis  of a  program  and  a forward  access  path $\alpha$,    the
membership problem $\alpha\$ \in \mathscr{L}(G)$ is undecidable.
\end{lemma} 

\begin{proof}
Given a Turing machine and an input string, the  proof   of
Lemma~\ref{lemma:grammar-undecidable} generates a grammar. 
We shall define a function
for   each  non-terminal $\nont{i}{c}$   introduced  in   this grammar
such that the
liveness  analysis  of  the function will  result
in a set of productions that includes the productions for
$\nont{i}{c}$ generated for the proof.
As an example, it  can be verified that the grammar
for the function shown in Figure~\ref{fig:turing-pgm}(b)    includes      the
productions shown in Figure~\ref{fig:turing-pgm}(a). Here $\nont{\mathit{all}}{}$   corresponds   to  the   demand
$\sigma_{\mathit{all}}$.
 
The body of the function corresponds  to the RHS of the productions
for $\nont{i}{c}$. Productions with the same  LHS non-terminal, but
with differing non-terminal on RHS can be generated by 
joining program fragments with $\SIF$.   If  there  are  $n$ such  functions,  then
\mainpgm\ creates  a $n$-way  branch and  inserts a  single call  to a
distinct function in each branch.

\begin{figure}[t]
  \begin{minipage}{.45\columnwidth}
    \begin{align*}
      \nont{i}{\acar} &\rightarrow \acar\nont{j}{\acar}\bcar\nont{\mathit{all}}{} \\
      \nont{i}{\acar} &\rightarrow \clazy\nont{j}{\acar}\bcar\nont{\mathit{all}}{} \\
      \nont{\mathit{all}}{} &\rightarrow \epsilon \mid \acar\nont{\mathit{all}}{} \mid
      \acdr\nont{\mathit{all}}{}
    \end{align*}
  \end{minipage}
  \begin{minipage}{.45\columnwidth}
    \input{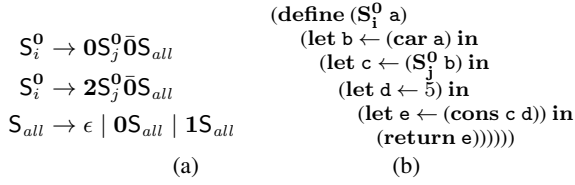}
  \end{minipage}\\
  \mbox{}\hfill(a)\hfill(b)\hfill\mbox{}
  \caption{A possible grammar generated by the proof of
    Lemma~\ref{lemma:grammar-undecidable} and a program to realize the
  grammar.}\label{fig:turing-pgm}
%\vspace*{-6mm}
\end{figure}

Notice that  each production  in Lemma~\ref{lemma:grammar-undecidable}
had a single non-terminal on the RHS. There are two characteristics of
the grammar produced from liveness-analysis  that are are relevant for
replaying the  earlier proof:  (a) If $\mathsf{S}  \rightarrow \beta_1
\mathsf{S'}  \beta_2$  was a  production  in  the earlier  proof,  the
grammar  generated from  liveness  analysis of the constructed program
will  have the  production
$\mathsf{S}      \rightarrow      \beta_1     \mathsf{S'}      \beta_2
\nont{\mathit{\mathit{all}}}{}$,  and (b)  other  productions are  of the  form
$\mathsf{S}       \rightarrow       \beta_1'\mathsf{S'}       \beta_2'
\nont{\mathit{all}}{}$.

It is clear that if the TM given as an instance of the halting problem
accepts the input string, then the earlier derivation can be replayed,
every   time   replacing   $\nont{\mathit{\mathit{all}}}{}$   in   the
sentential form by $\epsilon$. However, if  the TM does not accept the
input string, then every sentential form derived from the start symbol
would    have     at    least    one    non-terminal from the grammar
of Lemma~\ref{lemma:grammar-undecidable} that is different
from $\nont{\mathit{final}}{c}$.   Thus  $\$$  would  not  be
derivable from the grammar.
\end{proof}

%------------------------------------------------------------%
\begin{figure}[t!]
\includegraphics[width=\columnwidth]{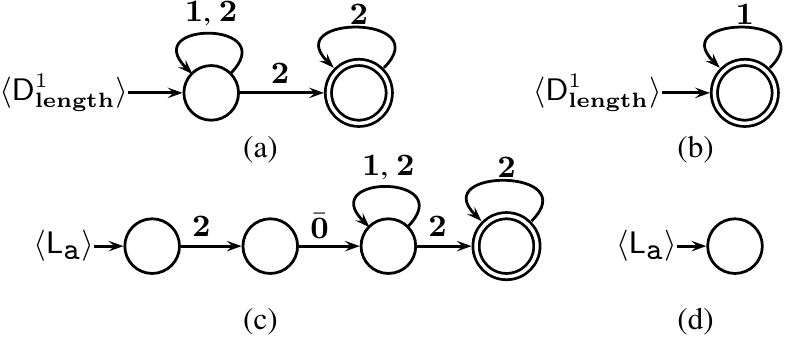}
\caption{(a) The  grammar rules for \var{\Df{\length}{1}}
  converted into an automaton, and (b) its DFA. The same for \var{\Lanv{\pa}{}}
  are in (c) and (d).}\label{fig:example-automata}
%\vspace*{-2mm}
\figrule
\end{figure}

%\subsection{A liveness-based GC scheme} \label{sec:live-clo}

%%\subsubsection{From CFGs to GC-ready DFAs}
%%\label{sec:NFA-approx}

We circumvent the problem of undecidability by over approximating the
CFG by non-deterministic finite state automata (NFA) using
%% The NFAs  are then simplified using the
%%$\hookrightarrow$ rules.   Finally the  simplified NFAs
%%are converted to DFAs.
%%
%% We use  the algorithm by 
Mohri  and Nederhof~\cite{mohri00regular} method.
%to approximate a CFG by a {\em strongly regular\/} grammar.  
For example, the grammar fragment for %the non-terminal
$\var{\Df{\length}{1}}$ after the Mohri-Nederhof transformation is:
 \begin{eqnarray*}
   \var{\Df{\length}{1}} &\rightarrow& \clazy\var{\Df{\length}{1'}} \mid
   \acdr\var{\Df{\length}{1}}
   \mid \clazy\var{\Df{\length}{1}}\\
   \var{\Df{\length}{1'}} &\rightarrow& \clazy\var{\Df{\length}{1'}}
   \mid \epsilon
 \end{eqnarray*}

The strongly regular grammar is converted  into a set of NFAs, one for
each $\var{\Lanv{x}{}}$.  The  $\hookrightarrow$ simplification is now
done on the NFAs by  repeatedly introducing $\epsilon$ edges to bypass
pairs  of consecutive  edges  labeled \bcar\acar\  or \bcdr\acdr\  and
constructing  the $\epsilon$-closure  till a  fixed point  is reached,
after which  the edges labeled  \bcar \  and \bcdr\ are  deleted.  
\cmt{ %% REMOVED FOR BLIND REVIEW
  The details of the  algorithm, its correctness and  termination proofs are
given  in~\cite{karkare07liveness,asati14lgc}. }
The
resulting  automaton   has  edges  labeled  with   \acar,  \acdr\  and
\clazy\ only.  In this  automaton, for every  edge labeled  \clazy, we
check if the source node of the edge has a path to a final 
state.  If it does, we mark  the source node as final. Finally, we
remove all the edges labeled \clazy\  and convert the automaton into a
deterministic    automaton.    These steps   effectively   implement    the
$\hookrightarrow$   simplification  rules   for   \bcar,  \bcdr,   and
\clazy\ to  obtain forward access paths.  While checking for  liveness during
garbage  collection, a  forward access path  is valid  only if  it can
reach a  final state.  Figure~\ref{fig:example-automata}(a)  shows the
NFA that is obtained from  the grammar for \var{\Df{\length}{1}}, and   the
final  DFA  is  shown in  Figure~\ref{fig:example-automata}(b).   This
expectedly  says that for a  demand  $\sigma_{\length}$,  the
liveness of  the argument of \length\  is $\acdr^{*}$  (the spine  of the  list is
traversed).  Similarly, Figure~\ref{fig:example-automata}(c) shows the
NFA     for     \var{\Lanv{\pa}{}}.      The     DFA in
Figure~\ref{fig:example-automata}(d)  does  not   accept  any  forward
paths, reflecting  the  lazy  nature  of  our  language.   Since
\length\ does not  evaluate the elements of  the argument list,
the  closure for  \pa\ is  never evaluated  and is  reclaimed whenever
liveness-based GC triggers beyond $\pi_9$.

  \SetStartEndCondition{ }{}{}%
  \SetKwProg{Pro}{procedure}{\string:}{}
  \SetKwProg{Fn}{function}{\string:}{}
  \SetKwFunction{Range}{range}%%
  \SetKw{KwTo}{in}\SetKwFor{For}{for}{\string:}{}%
  \SetKwIF{If}{ElseIf}{Else}{if}{:}{elif}{else:}{}%
  \SetKwFor{While}{while}{:}{fintq}%
  \AlgoDontDisplayBlockMarkers\SetAlgoNoEnd\SetAlgoNoLine%
  \SetKwFunction{Lgc}{$\mathsf{lgc}$}%
  \SetKwFunction{Copy}{$\mathsf{copy}$}%
  \SetKwFunction{LCopy}{$lgcCopy$}%
  \SetKwFunction{RCopy}{$rgcCopy$}%

\begin{algorithm}[t!]
  \Pro{\Lgc{}}
     {
       \For {each reference $\mathsf{ref}$ in root set}
            {$\mathsf{ref}$ = \Copy($\mathsf{ref}, \mathsf{init}(\mathsf{ref.dfa})$)\;}
            ${\mathsf{copyReferencesOnPrintStack}()}$\;  
     }
     \Fn{\Copy{$\mathsf{ref}, \mathsf{state}$}}
        {
          \eIf {$\mathsf{final}(\mathsf{state}$)}
             {
               $\mathsf{newRef} = \mathsf{dupHeapCell}(\mathsf{ref})$\;
               \If{$\mathsf{ref\!.cell}$($\mathsf{ref}$)
                 is a cons cell $\mathsf{(cons~arg_1~arg_2)}$}
                  {
                    {
                      $\mathsf{newRef}\!.\mathsf{arg_1}  = \Copy(\mathsf{arg_1}, \mathsf{next}(\mathsf{state}, 0))$\;
%                      $\mathsf{newRef}\!.\mathsf{arg_1}   =  \mathsf{newCar}$\;
                    }
                    {
                      $\mathsf{newRef}\!.\mathsf{arg_2} =  \Copy(\mathsf{arg_2},
                      \mathsf{next}(\mathsf{state} , 1))$\;
%                      $\mathsf{newRef}\!.\mathsf{arg_2}   = \mathsf{newCdr}$\;
                    }    
                  }
               \If
                  {$\mathsf{ref\!.cell}$ is a
                    closure cell, generically  $\mathsf{(binop~arg_1~arg_2)}$}
                   { 
                     $\mathsf{newRef}\!.\mathsf{arg_1} = \Copy(\mathsf{arg_1}, \mathsf{init}(\mathsf{arg_1.dfa}))$\;
                     $\mathsf{newRef}\!.\mathsf{arg_2} = \Copy(\mathsf{arg_2}, \mathsf{init}(\mathsf{arg_2.dfa}))$\;
                   }
             }
             {$\mathsf{newRef = ref}$}
\KwRet $\mathsf{newRef}$\;
        }
        \caption{Liveness-based garbage collection.  \label{algo:lgc-a}}
\end{algorithm}

\newcommand{\wdh}{.99\textwidth}
\newcommand{\LGCLine}{blue}
\newcommand{\RGCLine}{red}
\newcommand{\ReachLine}{black}
\newcommand{\UseLine}{light-blue}

\newcommand{\mycaption}{\caption{(Continued). Memory usage.  
The \RGCLine\ and the \LGCLine\ curves indicate the number of cons
cells  in  the  active   semi-space  for  RGC  and  LGC
respectively.  The \ReachLine\ curve represents the number of
reachable cells and the  \UseLine\ curve represents the
number  of  cells  that  are actually  live  (of  which
liveness analysis does a static approximation).  x-axis
is the time measured  in number of cons-cells allocated
(scaled down by factor $10^5$). y-axis is the number of
cons-cells (scaled down by $10^3$).}}
\begin{figure*}[p]
 \caption{Memory usage.  
The \RGCLine\ and the \LGCLine\ curves indicate the number of cons
cells  in  the  active   semi-space  for  RGC  and  LGC
respectively.  The \ReachLine\ curve represents the number of
reachable cells and the  \UseLine\ curve represents the
number  of  cells  that  are actually  live  (of  which
liveness analysis does a static approximation).  x-axis
is the time measured  in number of cons-cells allocated
(scaled down by factor $10^5$). y-axis is the number of
cons-cells (scaled down by $10^3$).}
\label{fig:memory-usage}
\addSubFig{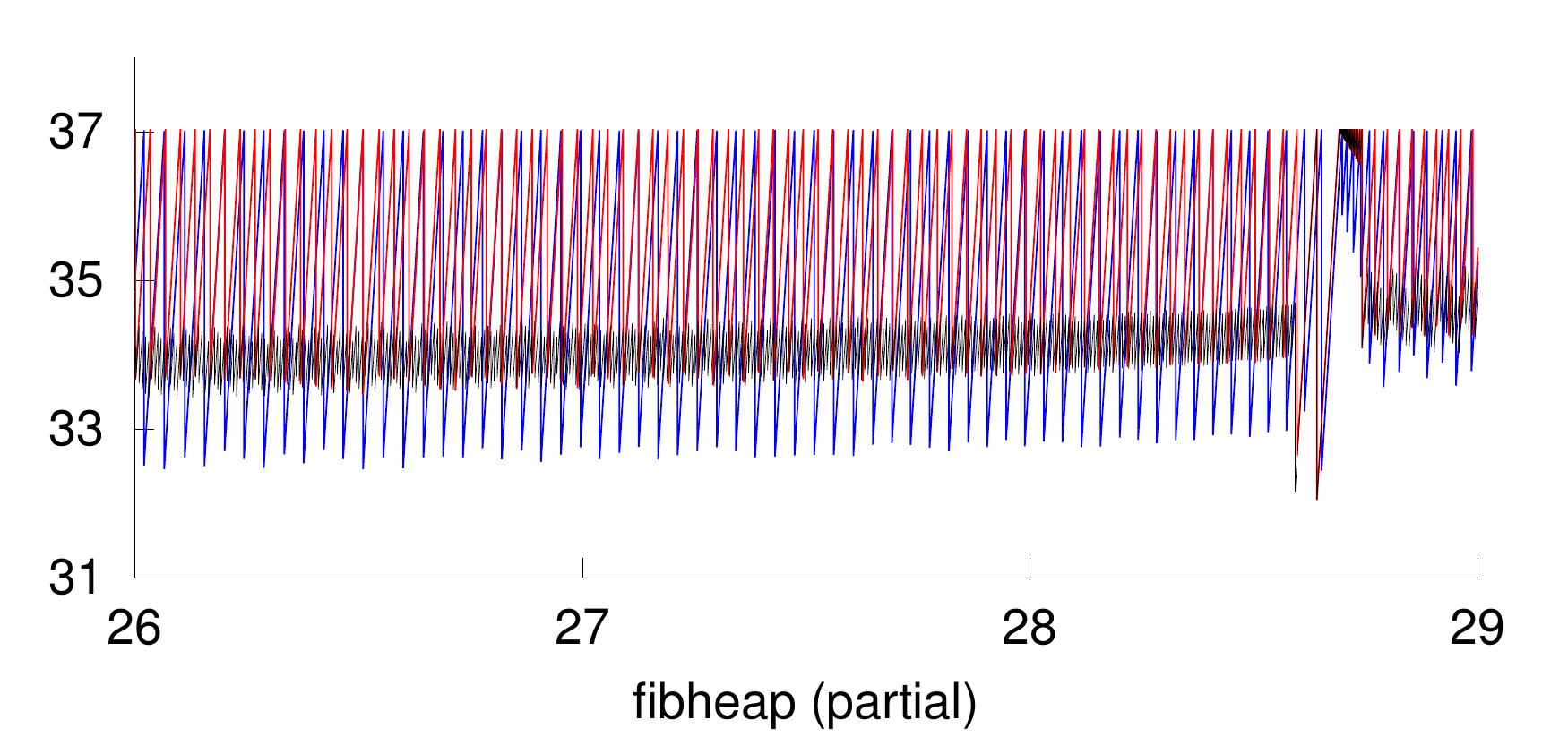}
\end{figure*}

\begin{figure*}[p]
  \ContinuedFloat\mycaption
  \addSubFig{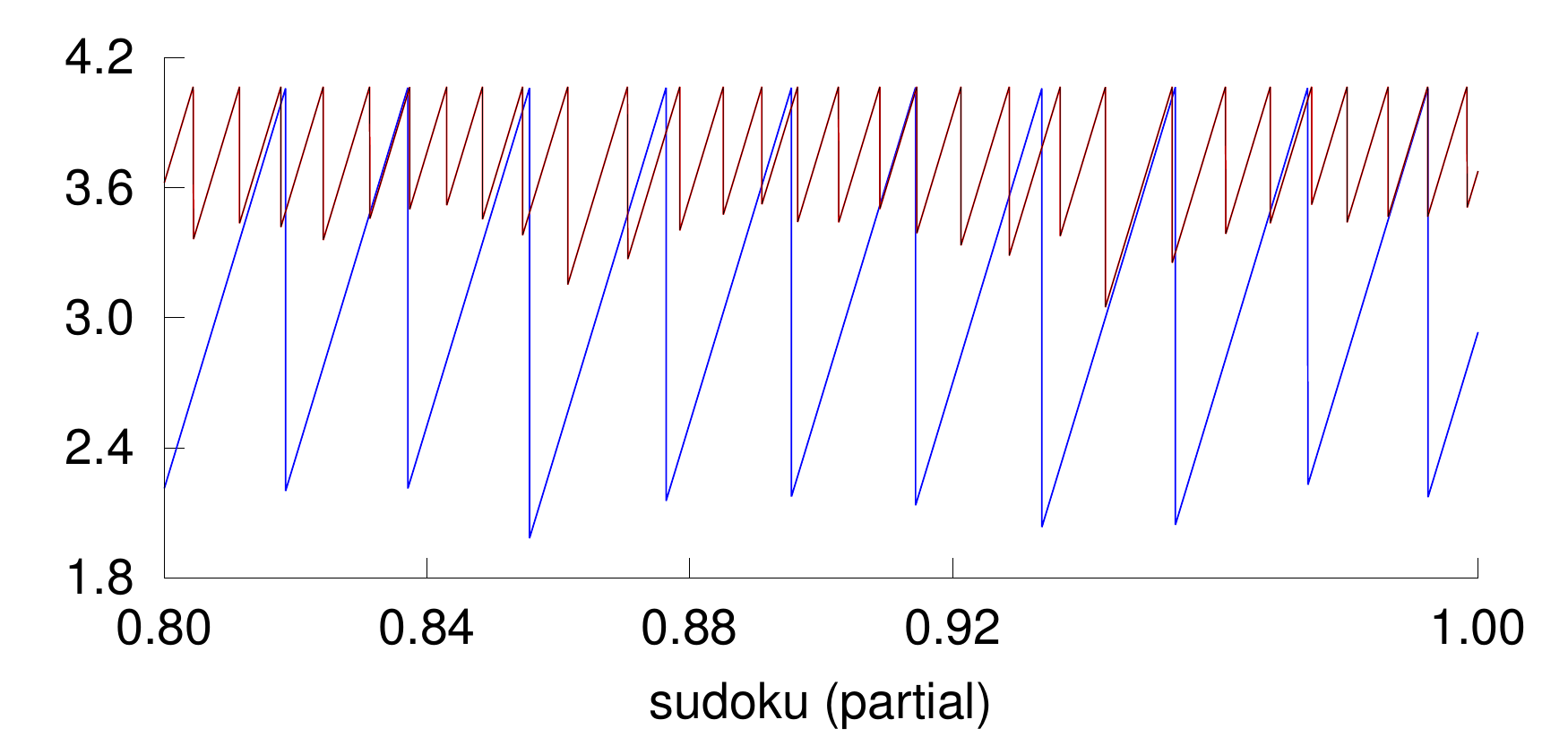}
\end{figure*}

\begin{figure*}[p]
  \ContinuedFloat\mycaption
  \addSubFig{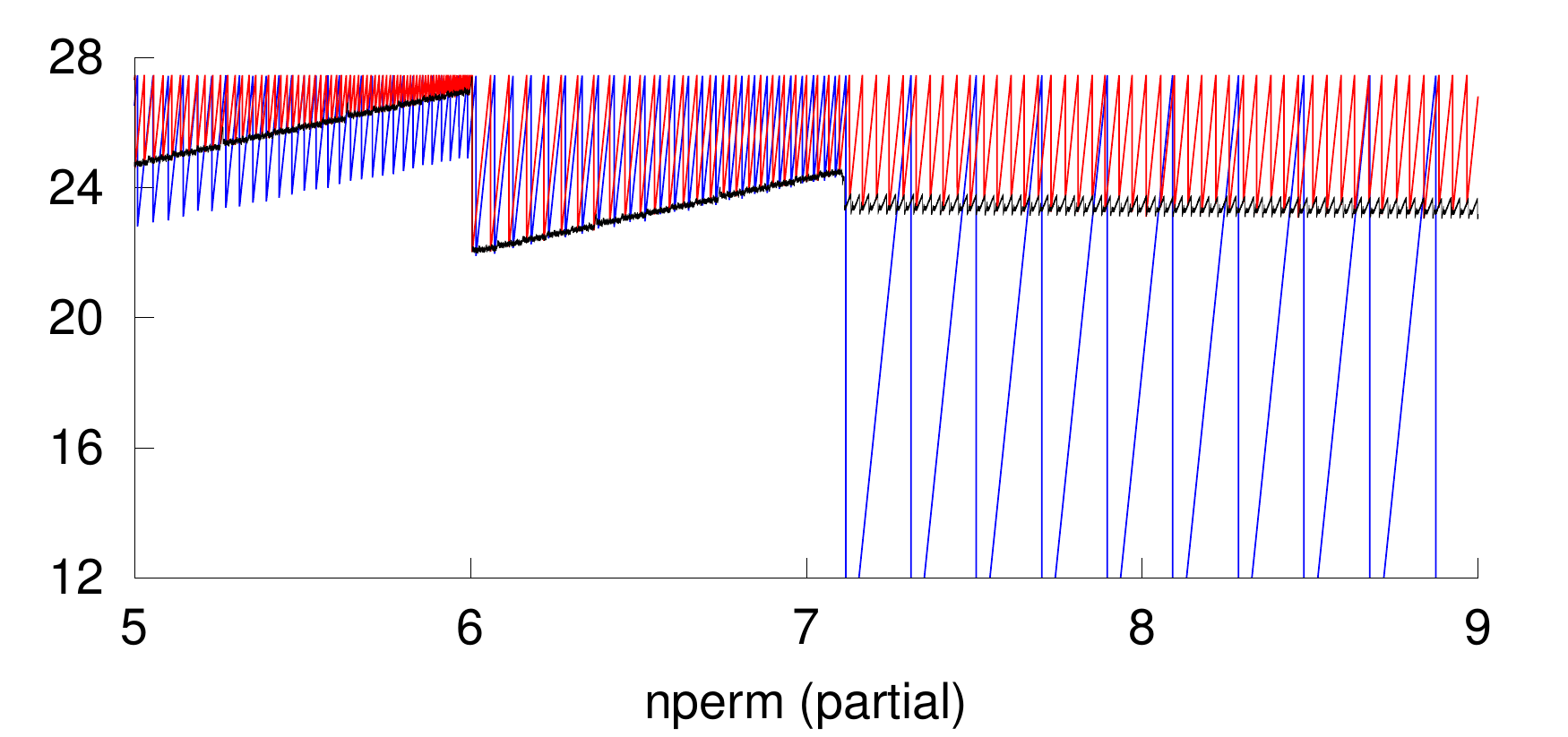}
\end{figure*}

\begin{figure*}[p]
  \ContinuedFloat\mycaption
  \addSubFig{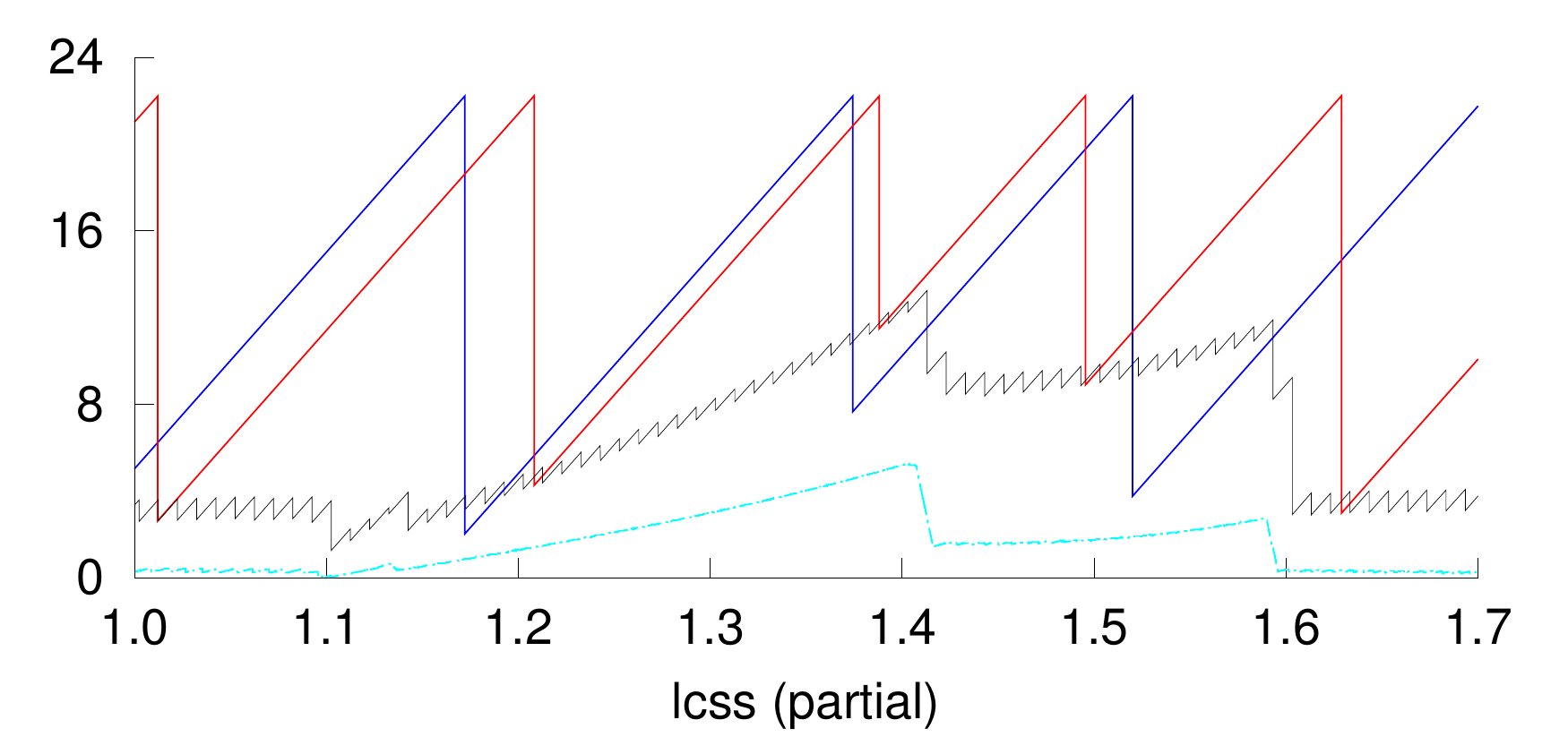}
\end{figure*}

\begin{figure*}[p]
  \ContinuedFloat\mycaption
  \addSubFig{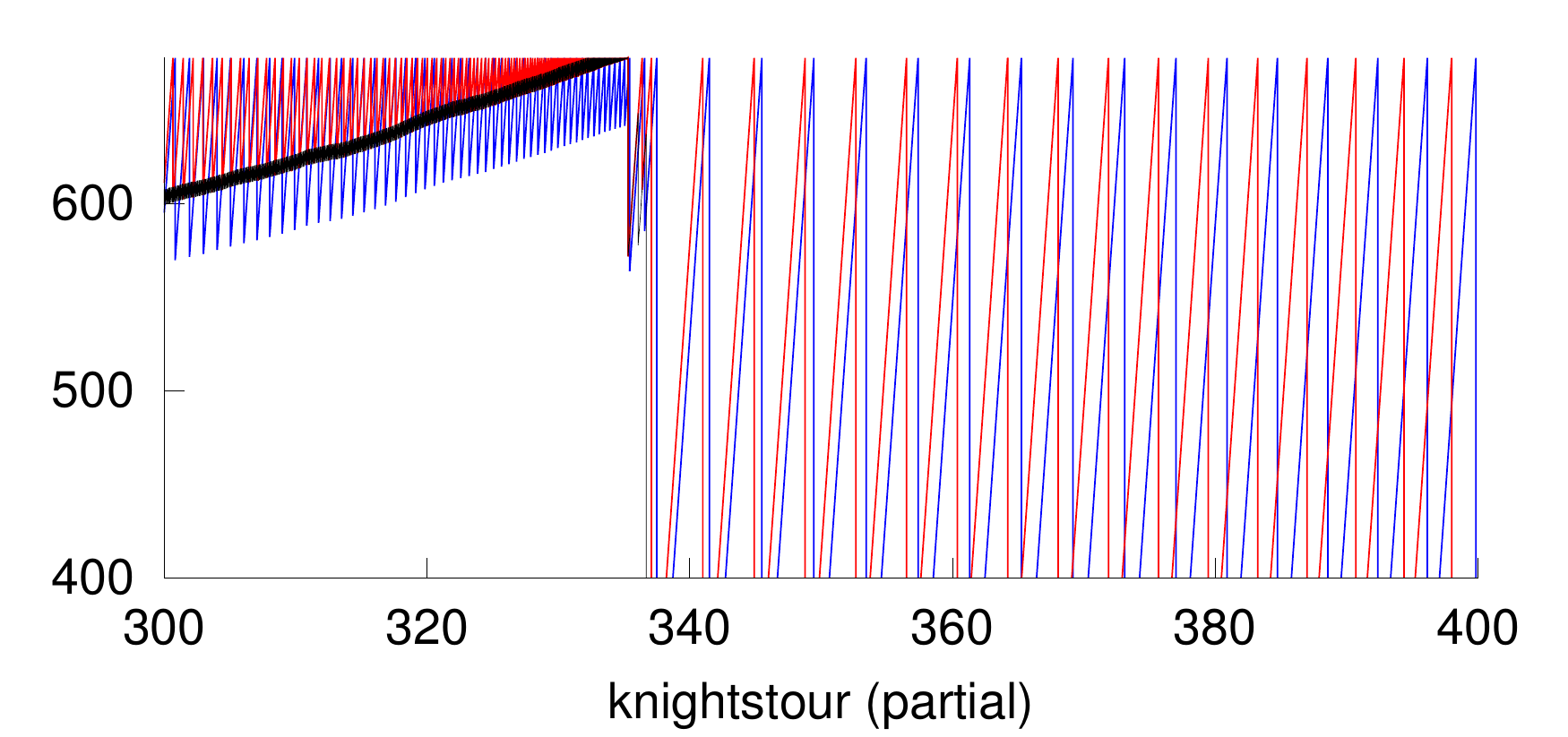}
\end{figure*}

\begin{figure*}[p]
  \ContinuedFloat\mycaption
  \addSubFig{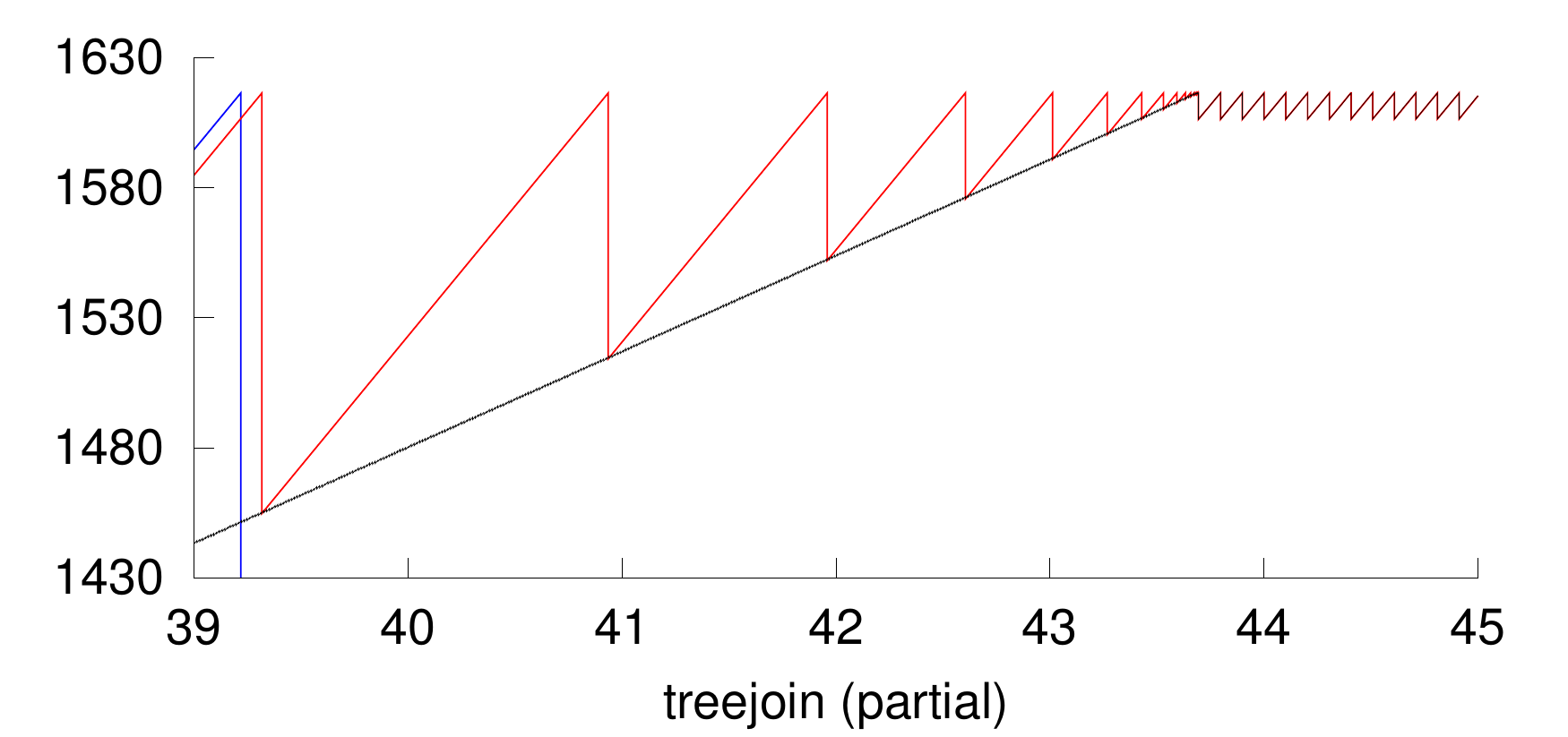}
\end{figure*}

\begin{figure*}[p]
  \ContinuedFloat\mycaption
  \addSubFig{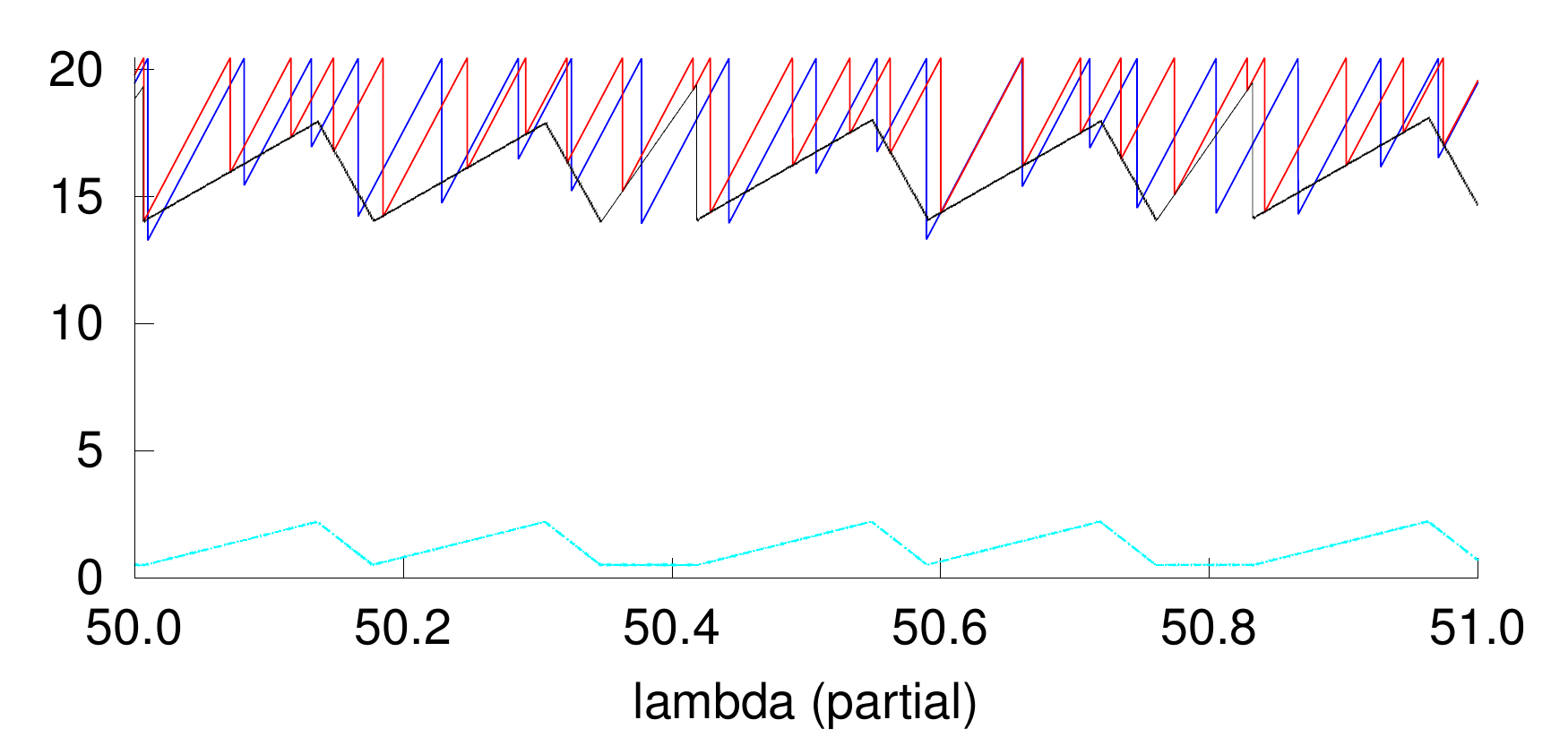}
\end{figure*}

\begin{figure*}[p]
  \ContinuedFloat\mycaption
  \addSubFig{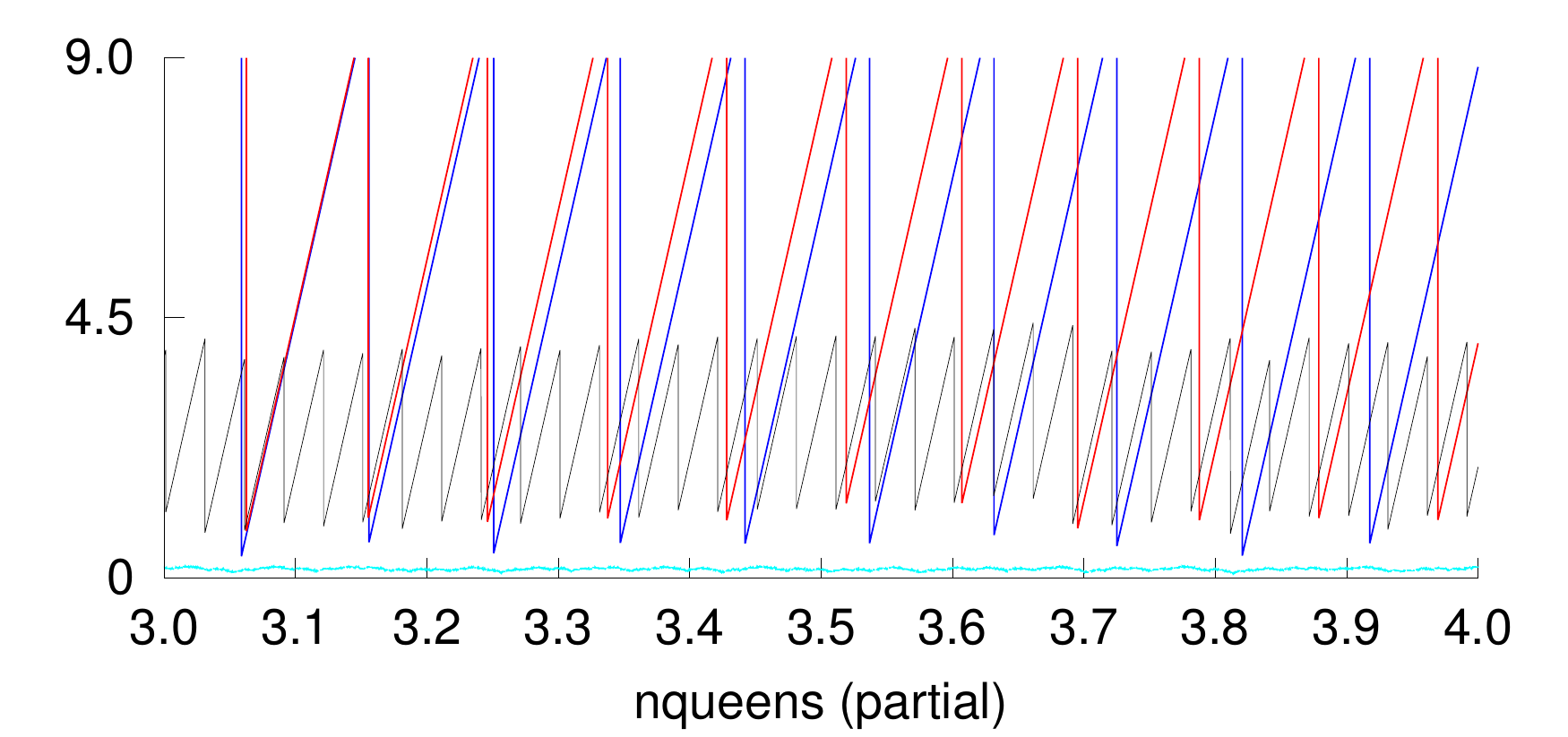}
\end{figure*}

\begin{figure*}[p]
  \ContinuedFloat\mycaption
   \begin{subfigure}{\textwidth}
     \centering
    \includegraphics[width=\wdh]{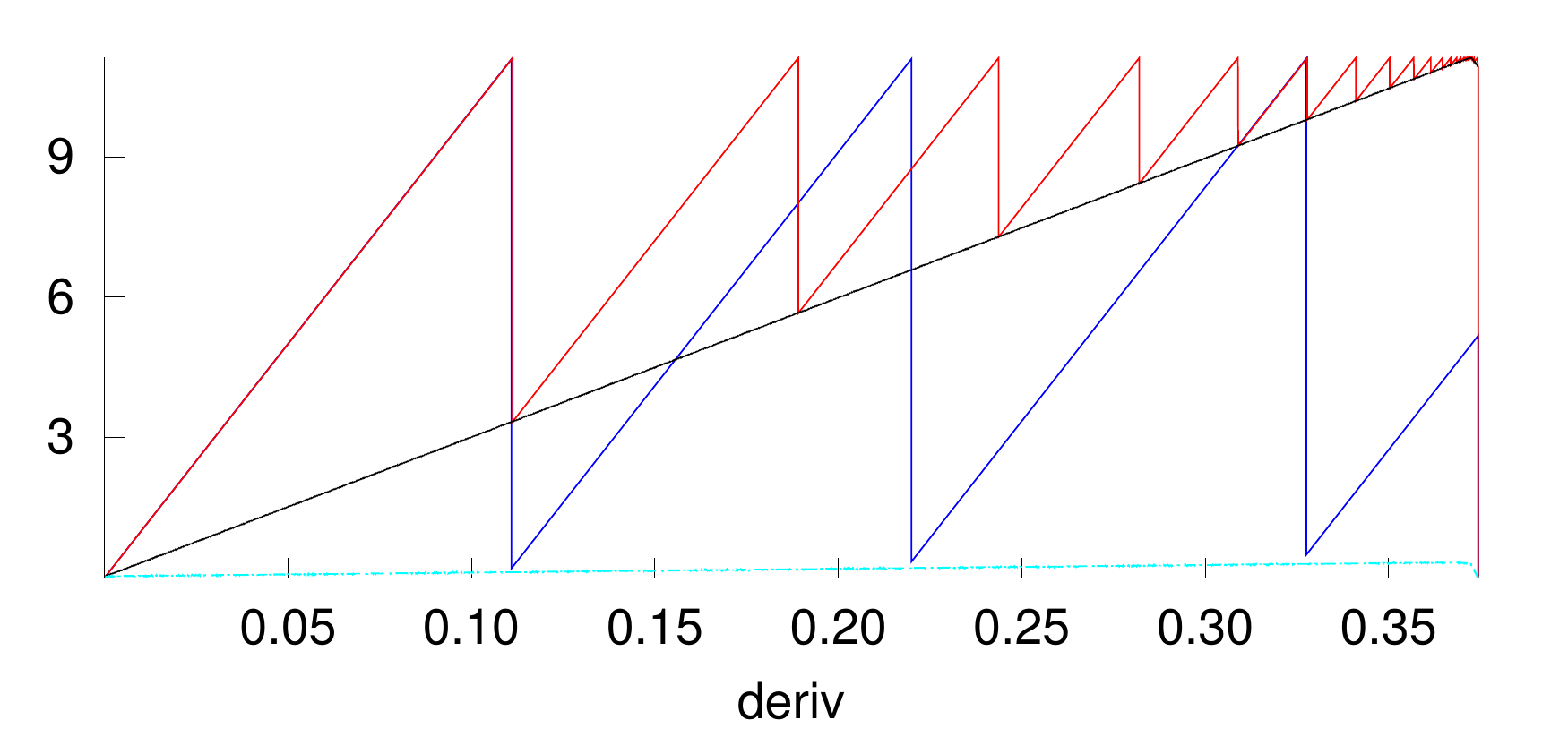}
  \end{subfigure}
   \begin{subfigure}{\textwidth}
     \centering
    \includegraphics[width=\wdh]{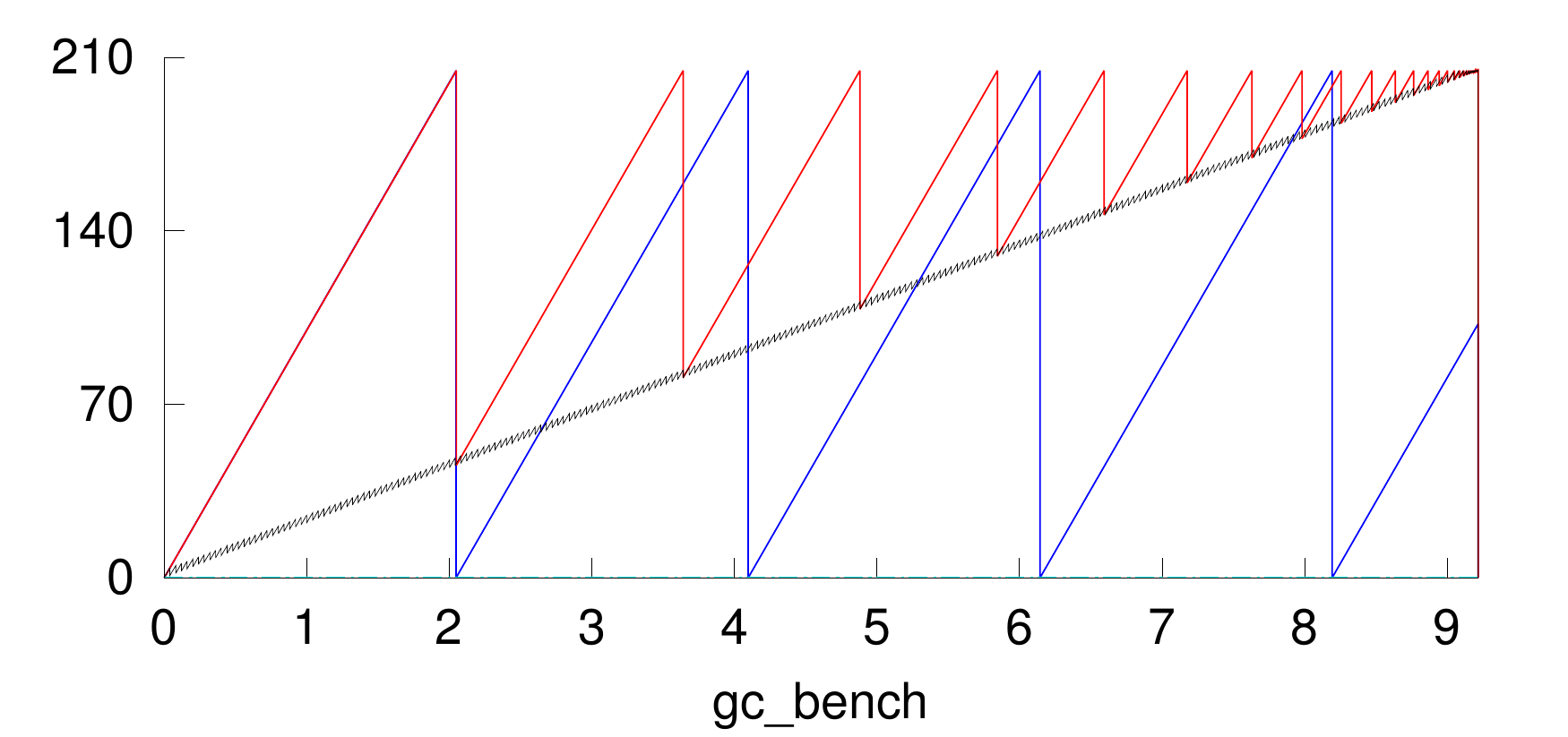}
  \end{subfigure}
\end{figure*}

\begin{figure*}[p]
  \ContinuedFloat\mycaption
  \begin{subfigure}{\textwidth}
     \centering
    \includegraphics[width=\wdh]{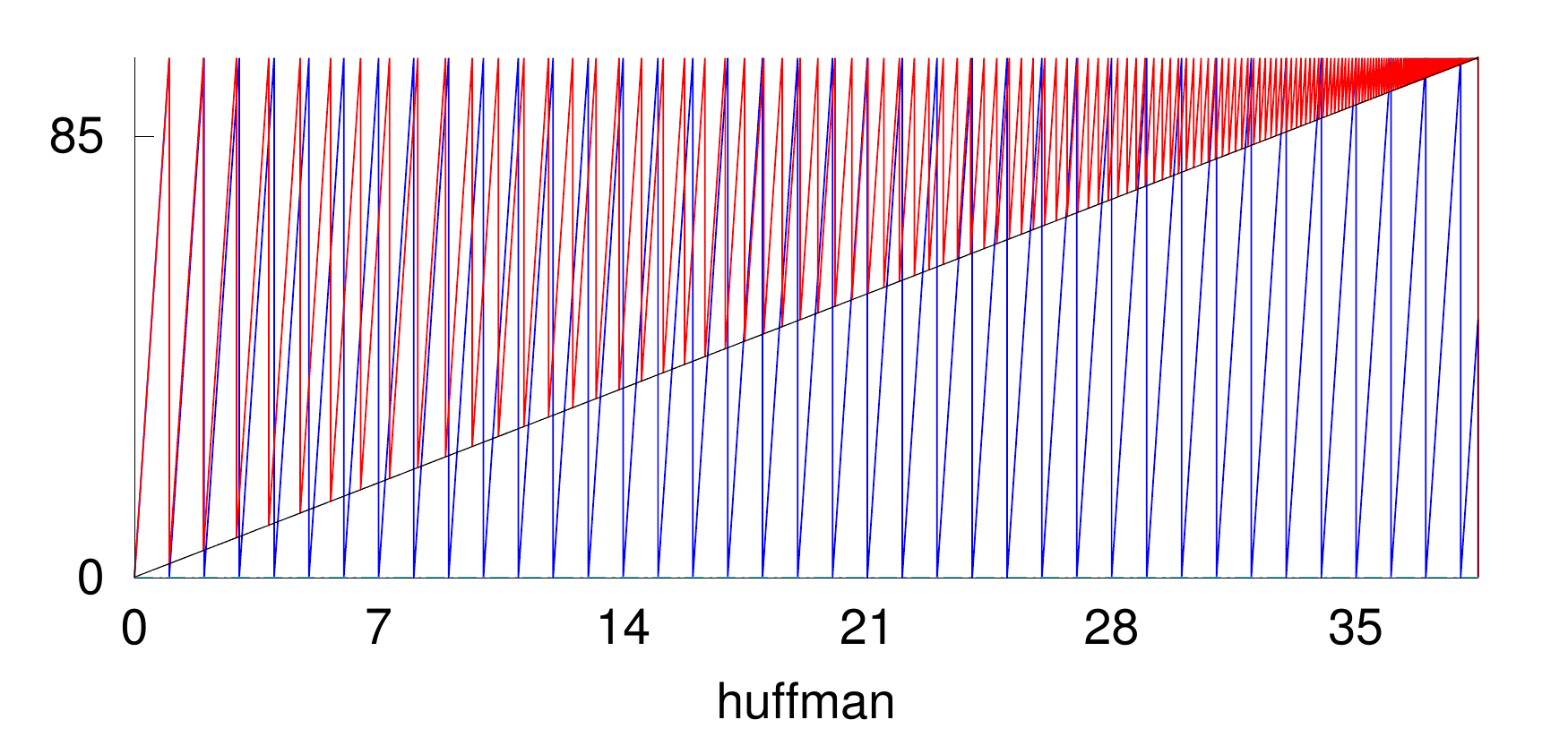}
  \end{subfigure}
  \begin{subfigure}{\textwidth}
     \centering
   \includegraphics[width=\wdh]{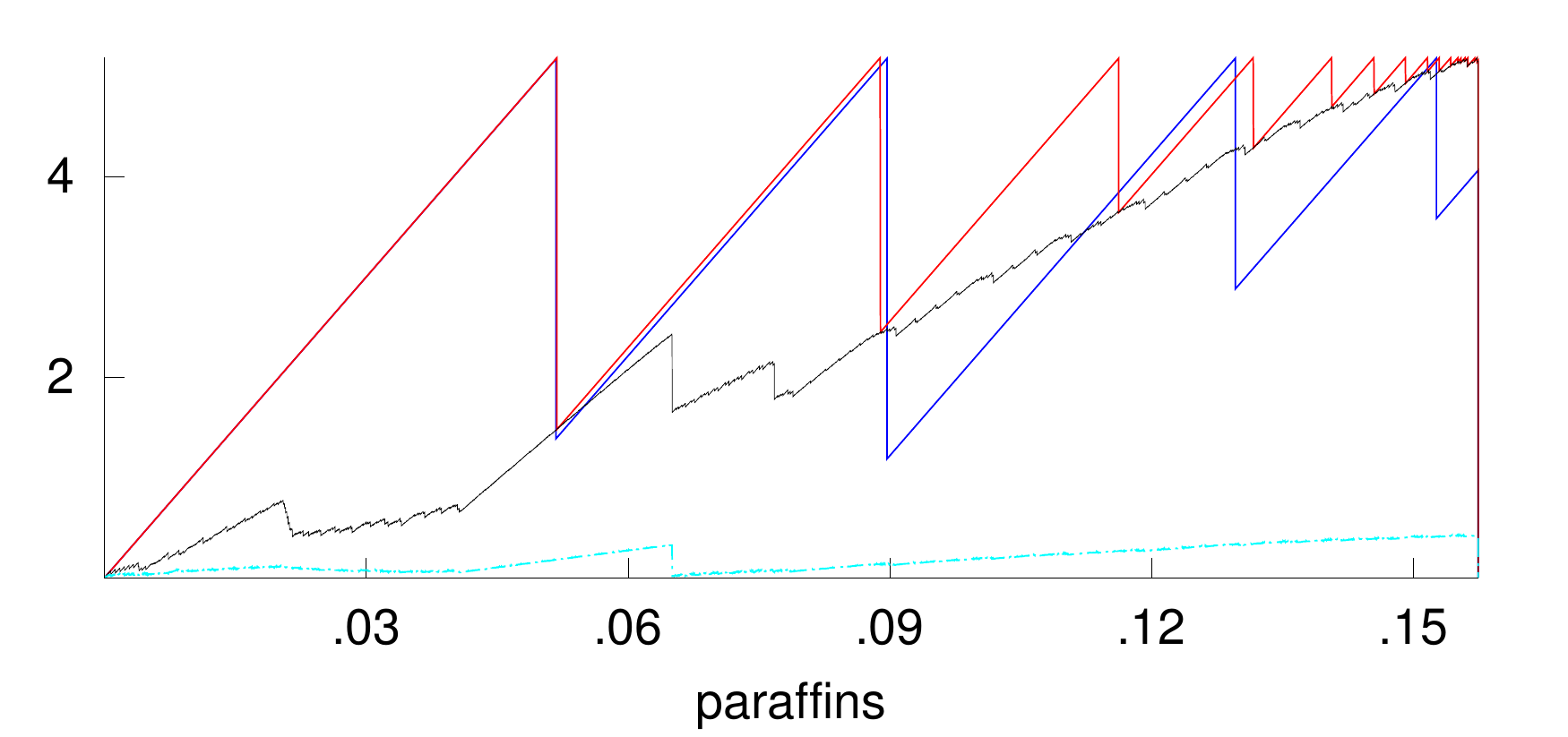}
  \end{subfigure} 
 \figrule
\end{figure*}

\section{The GC scheme}
\label{sec:GC-scheme}
In        the       liveness        analysis       described        in
Section~\ref{sec:liveness-analysis},   the  liveness   of  a   closure
variable is  derived from dependence  chains along all paths  from the
program point  where the closure  was created to  reachable evaluation
points.  Assume,  for the  sake of  concreteness, that  $\epath_1$ and
$\epath_2$ are two such evaluation points.  During GC,
we would like to use a more precise liveness, based on the actual paths
taken  during  execution.   Therefore   we  create  separate  liveness
automata for dependences  along paths to $\epath_1$  and $\epath_2$, in
addition to  automata for dependences  along paths to  both $\epath_1$
and $\epath_2$.  The closure carries  the liveness environment for its
free  variables (as  pointers  to automata,  one  for each  variable).
Initially the liveness  environment is based on  the dependences along
both  $\epath_1$   and  $\epath_2$.   However,  after   evaluating  an
\SIF\ condition, the liveness environments are updated to one based on
either   $\epath_1$  or   $\epath_2$,  so   that  subsequent   garbage
collections are based on more precise liveness information.

Based on  the above considerations,  we restrict the  possible garbage
collection points in a function body to the following:
\begin{enumerate}
\item We  statically over-approximate the memory required  to create
  the closures for each function body. On entering a function, if the
  available memory is  less than this requirement for  the function, a
  GC is triggered.
\item Since the evaluation of a  \SIF\ condition may trigger a 
  collection, after evaluating the condition   the available memory
  is  checked once  again against  a  revised estimate  of the  memory
  (based on  value of the condition)  required to execute the  rest of
  the program.  A GC is triggered if enough memory is unavailable.
\end{enumerate}

%% While  exploring  the
%% activation record of a  function, LGC uses the DFA
%% at  the  latest  program point  traversed  in  the
%% function's body.   Thus, if LGC is  invoked at the
%% program point $\pi$, it would use the liveness DFA
%% at $\pi$ itself for exploring  the root set of the
%% current   activation   record.   For   any   other
%% activation record in the stack, say for a function
%% ${\mathit  f}$   calling  ${\mathit  g}$   in  the
%% call-chain, the evaluation  point in ${\mathit f}$
%% which resulted in the  call to ${\mathit g}$ would
%% be used.

We shall call a  unit of allocatable memory as a {\em  cell}.  A cell can
hold    a    basic    value    ($\mathsf{bas}$),    the    constructor
\CONS\ $(\mathsf{cons~  arg_1~arg_2}$) or a closure.   The closure, in
turn,    can    be    one   of    ($\mathsf{unop~arg}$),    ($\mathsf{
  binop~arg_1~arg_2}$)  and  function application  ($\mathsf{f~arg})$.
Here each $\mathsf{  arg_i}$ is a reference to another  heap cell.  In
addition, the closure also carries  a pointer to a DFA (denoted
$\mathsf{arg_i.dfa_i}$) for each $\mathsf{arg_i}$.
Algorithm~\ref{algo:lgc-a}   describes   the   garbage
collection scheme.  Starting with the root set, each cell pointed by a
live reference  (i.e., whose associated  DFA state is final)  is copied
using $\mathsf{copy}$.   Copying a \CONS\ cells  just involves copying
the  cell  itself   and  conditionally  copying  the   \CAR\  and  the
\CDR\ fields after referring to the  next states of the DFA.  If
the reference points to a closure, then, as noted earlier, the closure
carries pointers  to the liveness  DFAs of its  arguments.  These
are  used to recursively  initiate copying of  the arguments.
Note that  the copying  strategy for  $\mathsf{(unop~arg_1~arg_2)}$ or
$\mathsf{(f~arg_1)}$ are similar to $\mathsf{(binop~arg_1~arg_2)}$ and
have not been shown.

\begin{table*}[t!]
\caption{Statistics for liveness analysis and garbage collection}
\label{tab:exp-results}
%\vskip -2mm
\centering
\renewcommand{\arraystretch}{.9}
\begin{tabular}{|l||c@{\ }|c@{\ }|c@{\ }|c@{\ }|c@{\ }|c@{\ }|c@{\ }|c@{\ }|c@{\ }|c@{\ }|c@{\ }|c@{\ }|} \hline
 \multicolumn{1}{|r||}{Program} & \rot{\tt fibheap} & \rot{\tt sudoku} & \rot{\tt nperm} & \rot{\tt paraffins} & \rot{\tt lcss} & \rot{\tt huffman} & \rotmore{\tt knightstour} & \rot{\tt nqueens} & \rot{\tt deriv} & \rot{\tt treejoin} & \rot{\tt lambda} & \rot{\tt gc\_bench} \\ \hline
 \#CFG Nonterminals & 621 & 1422 & 662 & 1174 & 642 & 499 & 660 & 404 & 328 & 615 & 669 & 390 \\ \hline
 \#CFG Rules & 1176 & 2009 & 866 & 1773 & 1206 & 818 & 883 & 643 & 468 & 1328 & 1088 & 450 \\ \hline
 \#DFA States & 1761 & 4283 & 1546 & 3346 & 1666 & 1414 & 1519 & 889 & 809 & 1803 & 1703 & 571 \\ \hline
 \#DFA Transitions & 2829 & 7690 & 2522 & 6086 & 2726 & 2528 & 2420 & 1170 & 1435 & 2797 & 2580 & 788 \\ \hline
 DFA Gen Time (sec) & 37.28 & 655.41 & 0.94 & 13.22 & 8.66 & 4.00 & 10.97 & 0.36 & 0.61 & 903.14 & 11.01 & 0.10 \\ \hline
\end{tabular}\\
 (a) Data for Liveness Analysis 

%\vskip -5mm
\newcommand{\rlratio}{$\frac{LGC}{RGC}$}
\newcommand{\rlratiotrue}{\raisebox{-2mm}{\rule{0pt}{6mm}}%
  $\frac{1.16\times LGC}{RGC}$}
\renewcommand{\arraystretch}{.6}
\begin{tabular}{|@{}c@{}|@{}r@{\ }|@{\ }r@{\ }|@{}r@{\ }|@{}r@{\ }|@{\ }r@{\ }| @{\ }r@{\ }|@{\ }r@{\ }|@{\ }r@{\ }|@{}r@{}| @{\ }r@{\ }|@{\ }r@{\ }|@{\ }r@{\ }|@{\ }r@{\ }|@{\ }r@{\ }|@{}r@{}|}
\hline
  &   \multicolumn{2}{@{}c@{\ }|}{\#Cells collected/GC}
  &   \multicolumn{2}{c@{\ }|}{\#Cells touched/GC}
  &   \multicolumn{2}{c@{\ }|}{\#GCs}
  &   \multicolumn{3}{c@{}|}{Peak Memory Required} %{Min. Cells Required}
  &   \multicolumn{3}{c@{\ }|}{GC time (sec)} %& Ratio 
  &   \multicolumn{3}{@{}c@{}|}{Total Exec time (sec)}\\
\cline{2-16}
{Program}&RGC&LGC&RGC&LGC&RGC&LGC&RGC&LGC&\rlratiotrue&RGC&LGC&\rlratio&RGC&LGC&\rlratio\\
\hline
\hline

\rotmore{\tt fibheap} & 3466.2 & 4164.5 & 33576.6 & 50957.9 & 1333 & 1108 & 37043 & 37043 & 1.16 & 1.46 & 20.00 & 13.70 & 13.82 & 33.23 & 2.40
\\ \hline

\rotmore{\tt sudoku} & 931.3 & 2328.7 & 3134.6 & 2950.2 & 179 & 72 & 4066 & 2960 & 0.84 & 0.01 & 0.04 & 4.64 & 0.11 & 0.16 & 1.5
\\ \hline

\rotmore{\tt nperm} & 4684.4 & 14212.5 & 22743.2 & 23127.1 & 710 & 235 & 27428 & 25343 & 1.07 & 0.24 & 1.24 & 5.20 & 1.80 & 3.39 & 1.88
\\ \hline

\rotmore{\tt paraffins} & 661.7 & 2920.5 & 4522.6 & 2933.0 & 16 & 4 & 5185 & 3733 & 0.83 & 0.00 & 0.00 & 2.01 & 0.01 & 0.02 & 1.2
\\ \hline

\rotmore{\tt lcss} & 8064.7 & 18268.7 & 14177.9 & 4604.4 & 30 & 14 & 22243 & 16296 & 0.84 & 0.01 & 0.02 & 3.02 & 0.13 & 0.18 & 1.36
\\ \hline

\rotmore{\tt huffman} & 10533.8 & 100010.0 & 89536.1 & 88.6 & 356 & 38 & 100070 & 72 & 0.00 & 0.64 & 0.01 & 0.02 & 2.70 & 2.55 & 0.94
\\ \hline

\rotmore{\tt knightstour} & 179155.0 & 312454.0 & 498645.0 & 534562.0 & 529 & 304 & 677800 & 642303 & 1.09 & 6.24 & 64.43 & 10.32 & 55.29 & 134.35 & 2.43
\\ \hline

\rotmore{\tt nqueens} & 2607.4 & 9529.1 & 7493.4 & 829.0 & 3345 & 916 & 10101 & 1082 & 0.12 & 0.36 & 0.18 & 0.49 & 4.68 & 5.87 & 1.25
\\ \hline

\rotmore{\tt deriv} & 854.6 & 10755.3 & 10269.3 & 420.3 & 31 & 3 & 11124 & 589 & 0.06 & 0.00 & 0.00 & 0.05 & 0.05 & 0.05 & 0.9
\\ \hline

\rotmore{\tt treejoin} & 50284.2 & 936525.0 & 1566250.0 & 700756.0 & 116 & 5 & 1616533 & 887005 & 0.63 & 3.84 & 1.90 & 0.49 & 6.66 & 5.50 & 0.82
\\ \hline

\rotmore{\tt lambda} & 7271.7 & 8448.4 & 13194.2 & 29156.3 & 775 & 667 & 20466 & 18169 & 1.02 & 0.17 & 4.70 & 28.48 & 2.77 & 8.32 & 3.00
\\ \hline

\rotmore{\tt gc\_bench} & 14932.7 & 204774.0 & 189880.0 & 33.2 & 48 & 4 & 204813 & 72 & 0.00 & 0.11 & 0.00 & 0.00 & 1.45 & 1.22 & 0.84
\\ \hline

\end{tabular}\\[1mm]
(b) Comparing RGC with LGC. Note that the size of an LGC cell is 1.16 times the size of an RGC cell, Total Exec time includes GC time.

%\vskip -7mm
\end{table*} 

%\subsection{Garbage collection for references on print stack}
The evaluation of  the top-level expression in a program  is driven by
a printing function~(Section
\ref{sec:semantics}).  
%% Printing of  atomic values is simple  and does
%% not  trigger a  GC.   In case  of  \CONS\ cells,  the
%% \CAR\ is  first printed followed  by \CDR.   In a lazy  language, both
%% \CAR\ and \CDR\  could be closures requiring evaluation,  and this may
%% trigger  more GCs.  When  this happens.   we have  to
%% consider any  references that  might be  on the  print stack  and copy
%% them.
We extend liveness-based  garbage  collection  to
this function.
\section{Experimental evaluation}
\label{sec:experiments}
Our experimental setup  consists of an interpreter
for  our  language,  a liveness  analyzer,  and  a
single generation copying  collector.  The garbage
collector can be configured to work on the
basis  of  reachability  (RGC  mode)  or  use  
liveness  DFAs  (LGC  mode).  

Our benchmark  consists of programs taken  from nofib~\cite{nofib} and
other repositories for functional programs~\cite{PLT-Scheme, gc_bench,
  huffman-sicp}.     We    ran    the   experiments    on    8    core
Intel\textsuperscript{\textregistered}        Core\textsuperscript{TM}
i7-4770 3.40GHz CPU  having 8192KB L2 cache, 16GB RAM,  running 64 bit
Ubuntu 14.04.

%% The  process  of  liveness-based  garbage  collection  involves  going
%% through  each  activation  record  on the  stack  and  exploring  each
%% variable in the current activation  record (root set). The liveness of
%% each variable is  determined using the program point  and the variable
%% name. The liveness  of all variables are stored as  DFA. These DFA are
%% then encoded as a table. A variable  is live only if an entry is found
%% in  this table.   All  cells that  are  live are  copied  to the  live
%% semi-space. In case  of \CONS\ cells the \CAR\ and  \CDR\ pointers are
%% chased  and if  they are  live  the copied  \CONS\ cell  will get  the
%% updated addresses of its \CAR\ and \CDR\ fields.

\subsection{Results}

The statistics  related to  liveness analysis  and DFA  generation are
shown in Table~\ref{tab:exp-results}(a). We  observe that the analysis
of  all  programs  except \verb@treejoin@  and  \verb@sudoku@  require
reasonable time.   The bottleneck in  our analysis  is the NFA  to DFA
conversion with  worst-case exponential behaviour. However,  since the
analysis has  to be  done only  once and its  results can be cached and
re-used, the time spent in analysis may be considered acceptable.

Table~\ref{tab:exp-results}(b) compares  GC statistics
for RGC and  LGC.  We report the number of  GC events,
average  number of  cells  reclaimed per  GC,  average
number of cells  touched per GC and the  total time to
perform all collections.   It is no surprise that the  number of cells
reclaimed per garbage  collection is higher and the  number of garbage
collections  lower  for  LGC.   The  cost of  LGC  is  higher  garbage
collection time, which increases the  overall execution time even with
reduced number of  collections.  However, the execution time  of LGC is
still comparable for most benchmarks (slowdown within 5X of
RGC in most cases) and better  for 3 benchmarks (2X speedup in the
best case).  Note that  \verb@gc_bench@~\cite{gc_bench} is a synthetic
benchmark that allocates complete binary
trees  of various  sizes that  are  never used  by the  program. As  a
result,  the benchmark  highly favours  LGC.  The  benchmark has  been
included  for  completeness,  and  we  do  not  consider  its  numbers
as being representative of real programs.

Memory   usage    graphs   for    the   benchmarks   are    shown   in
Figure~\ref{fig:memory-usage}.  In  all the  programs we can  see that
the curve  corresponding to  LGC (\LGCLine\ line)  dips below  the RGC
curve (\RGCLine\ line)  during GC.  The graphs also  include the curve
for reachable cells (\ReachLine) and live cells (\UseLine). These were
obtained by forcing  RGC to run at very high  frequency. The curve for
live  cells were  obtained by  recording  heap access  times and  post
processing the data at  the end of the program. Note  that the size of
an LGC  cell is 1.16 times  the size of  a RGC cell as  it potentially
might have  to store references  to liveness DFA of  closure arguments
(if the cell is a closure).

As  demonstrated by  the gap  between the  \RGCLine\ and  the \UseLine\
lines, a  large number of  cells which are  unused by the  program are
still copied  during RGC. LGC does  a much better job  of closing this
gap but still falls short of the  precision achieved by LGC in case of
eager languages ~\cite{asati14lgc}.  A major source of inefficiency in
LGC is  multiple traversals  of already copied  heap cells.  Since LGC
does not mark the heap cells after the first visit, the same cells can
be  repeatedly  visited  with   different  liveness  states.  We  have
mitigated  this  problem  by   implementing  a  heuristic  which  avoids 
revisiting closures and function arguments more than once. 
%% Since the 
%% liveness automata for a closure argument encompasses all possible uses 
%% of the closure, multiple traversals for closures can be safely avoided.

%\vspace*{-2mm}
\section{Related work} 
\label{sec:relatedwork}
The impact of  liveness on the effectiveness of  GC is
investigated in~\cite{Hirzel}.   They observe that  liveness can
significantly  impact   garbage  collection,  but  only   when  it  is
interprocedural.  As far as memory requirement is concerned, our paper
demonstrates this observation.

There have been  several attempts to use liveness  analysis to improve
GC for imperative  languages.~\cite{khedker07heap} presents a liveness
analysis and uses  the results for inserting  nullifying statements in
Java  programs. In~\cite{ran.shaham-sas03},  temporal properties  like
liveness  are checked  against an  automaton modeling  heap accesses.
Both these approaches are intraprocedural in scope.

In the space of functional languages, there are: rewriting methods
such as
deforestation~\cite{wadler88deforest,gill93ashort,chitil99deforest},
sharing analysis based reallocation~\cite{jones89compile}, region
based analysis~\cite{tofte98region}, and insertion of compile-time
nullifying statements~\cite{inoue88analysis, lee05static}.  All compile-time
marking approaches rely on an efficient and precise alias analysis and
cannot provide significant improvement in its absence. The only work
in the space of lazy languages seems to be~\cite{Hamilton}
which touches upon only basic techniques of compile-time garbage
marking, explicit deallocation and destructive allocation. An
interesting approach suggested in~\cite{HofmannJ03} is to annotate the
heap usage of first-order programs through linear types.  The
annotations are then used to serve memory requests through
re-allocation.  However, this requires the user to write programs in a
specific way.
Safe-for-space~\cite{appel.cps}              approaches~\cite{Clinger,
  shao00efficient}  reduce the  amount of  heap used  by a  program by
allocating   closures    in   registers   and   through    tail   call
optimizations. However, these approaches take care of only part of the
problem addressed  by our  analysis as the  program can  still contain
unused objects and closures that are reachable.

Simplifiers~\cite{ONeill}  are  abstractly  described  as  lightweight
daemons that  attach themselves to  program data and,  when activated,
improve the  efficiency of the  program. Our liveness-based GC  can be
seen as  an instance  of a  simplifier which  is tightly  coupled with
garbage  collectors.   The approach  that  is  closest to  the  method
described  in  this  paper  is the  liveness-based  garbage  collector
implemented  in~\cite{karkare07liveness,asati14lgc} and  address eager
languages. We extend this to handle lazy evaluation and closures.
 
%% A similar  approach has been suggested  in~\cite{ONeill} where they
%% suggest   augmenting  GC   with   extra   information  using   {\it
%% simplifiers}.

%% Previous attempts  to increase the  space efficiency
%% of functional programs  by additional reclamation of
%% memory fall  in two broad categories.  In the first,
%% the  program   itself  is  instrumented   to  manage
%% reclamation and reallocation without  the aid of the
%% GC.   Such attempts  include: sharing
%% analysis  based  reallocation~\cite{jones89compile},
%% deforestation
%% techniques~\cite{wadler88deforest,gill93ashort,chitil99deforest},
%% methods based on linear logic~\cite{hofmann00linear}
%% and region analysis~\cite{tofte98region}.  Closer to
%% our  approach, there  are  methods  that enable  the
%% garbage      collector      to     collect      more
%% garbage~\cite{inoue88analysis,lee05static}        by
%% explicitly  nullifying pointers  that are  not live.
%% However,  the nullification,  done at  compile time,
%% requires sharing  (alias) analysis.  Our  method, in
%% contrast, does  not require  sharing because  of the
%% availability of the heap  itself at runtime.  To the
%% best of our knowledge, this  is the first attempt at
%% liveness-based  marking of  the heap  during garbage
%% collection.

%% \begin{figure}[t]
%% \centerline{\epsfig{file=lambda.pdf, height=4cm, width=7cm}}
%%  \caption{Memory  usage and garbage collection pattern of  {\tt
%% lambda}.}
%% \label{fig:memory-usage-lambda} \figrule
%% \end{figure}
%\vspace*{-2mm}
%\vskip -5mm
\section{Future work and conclusions}
\label{sec:conclusion}
We have  extended liveness-based  GC  to lazy
  languages  and shown  its benefit  for a  set of  benchmark
  programs.  We  defined a liveness analysis  of programs manipulating
  the  heap and  proved  it  correct with  respect  to a  non-standard
  semantics that served as a specification of liveness.  The result of
  the analysis  is a  set of grammars,  whose membership  question was
  shown to be undecidable. The  grammars are therefore approximated by
  DFAs and used by the garbage collector to improve collection.

In addition to  evaluated data, our collector  also reclaims closures.
For this, we  had to modify the standard run-time  representation of
closures  to carry  liveness  of its free  variables and  periodically
update the liveness during execution  time.  As expected, our garbage
collector  reclaims more  garbage and  reduces memory  requirement of
programs.  An  additional benefit is  that in  spite of using  a more
expensive collector, the execution times remains comparable in
most cases, and even improves for some programs.

The graphs show that even in lazy languages there are large amount of
dead cells that can be collected early. In spite of collecting
closures there still exists a gap between actual and perceived
livenesses.  Language features that could be introduced are
higher-order functions and recursive lets. We believe that functions
can be used to model recursive lets and this should be a simple
extension to our liveness analysis.  For higher-order we intend to use
firstification techniques~\cite{Mitchell:2009}.

  Orthogonally,   we   plan  to   improve   the   efficiency  of   the
  liveness-based garbage  collector using heuristics such  as limiting
  the depth of DFA, merging  nearly-equivalent states and using better
  representation (for example BDDs~\cite{Bryant86})  and algorithms for
  automata manipulation.  We also  need to investigate the interaction
  of liveness with  other collection schemes, such  as incremental and
  generational collection.   It might  be interesting  to use  a mixed
  mode GC  scheme which allows the  costs of LGC to  be amortized over
  several runs of RGC. In summary, we need to investigate ways to make
  liveness-based GC attractive for practical collectors.

%\clearpage
\balance
\bibliography{fun_hra}{}

\begin{thebibliography}{10}

\bibitem{gc_bench}
An artificial garbage collection benchmark.
\newblock \url{http://www.hboehm.info/gc/gc_bench.html}, Nov 2015.
\newblock (Last Accessed).

\bibitem{huffman-sicp}
Huffman encoding trees.
\newblock \url{https://mitpress.mit.edu/sicp/full-text/sicp/book/node41.html},
  Nov 2015.
\newblock (Last Accessed).

\bibitem{nofib}
{NoFib: Haskell Benchmark Suite}.
\newblock \url{http://git.haskell.org/nofib.git}, Nov 2015.
\newblock (Last accessed).

\bibitem{PLT-Scheme}
{PLT Scheme Benchmark Suite}.
\newblock \url{http://svn.plt-scheme.org/plt/trunk/collects/tests/}, Nov 2015.
\newblock (Last accessed).

\bibitem{appel.cps}
A.~W. Appel.
\newblock {\em Compiling with Continuations}.
\newblock Cambridge University Press, 2006.

\bibitem{asati14lgc}
R.~Asati, A.~Sanyal, A.~Karkare, and A.~Mycroft.
\newblock Liveness-based garbage collection.
\newblock In {\em Compiler Construction - 23rd International Conference, {CC}
  2014, Held as Part of the European Joint Conferences on Theory and Practice
  of Software, {ETAPS} 2014, Grenoble, France, April 5-13, 2014. Proceedings},
  pages 85--106, 2014.

\bibitem{Bryant86}
R.~E. Bryant.
\newblock Graph-based algorithms for boolean function manipulation.
\newblock {\em {IEEE} Trans. Computers}, 35(8):677--691, 1986.

\bibitem{chakravarty03perspective}
M.~M.~T. Chakravarty, G.~Keller, and P.~Zadarnowski.
\newblock A functional perspective on {SSA} optimisation algorithms.
\newblock {\em Electr. Notes Theor. Comput. Sci.}, 82(2):347--361, 2003.

\bibitem{chitil99deforest}
O.~Chitil.
\newblock Type inference builds a short cut to deforestation.
\newblock In {\em Proceedings of the fourth {ACM} {SIGPLAN} International
  Conference on Functional Programming {(ICFP} '99), Paris, France, September
  27-29, 1999.}, pages 249--260, 1999.

\bibitem{Clinger}
W.~D. Clinger.
\newblock Proper tail recursion and space efficiency.
\newblock In {\em Proceedings of the {ACM} {SIGPLAN} '98 Conference on
  Programming Language Design and Implementation (PLDI), Montreal, Canada, June
  17-19, 1998}, pages 174--185, 1998.

\bibitem{Fairbairn1987}
J.~Fairbairn and S.~Wray.
\newblock {\em Functional Programming Languages and Computer Architecture:
  Portland, Oregon, USA, September 14--16, 1987 Proceedings}, chapter Tim: A
  simple, lazy abstract machine to execute supercombinators, pages 34--45.
\newblock Springer Berlin Heidelberg, Berlin, Heidelberg, 1987.

\bibitem{gill93ashort}
A.~J. Gill, J.~Launchbury, and S.~L. {Peyton Jones}.
\newblock A short cut to deforestation.
\newblock In {\em Proceedings of the conference on Functional programming
  languages and computer architecture, {FPCA} 1993, Copenhagen, Denmark, June
  9-11, 1993}, pages 223--232, 1993.

\bibitem{Hamilton}
G.~W. Hamilton.
\newblock Compile-time garbage collection for lazy functional languages.
\newblock In {\em Proceedings of the International Workshop on Memory
  Management}, IWMM '95, pages 119--144, London, UK, UK, 1995. Springer-Verlag.

\bibitem{Hirzel}
M.~Hirzel, A.~Diwan, and J.~Henkel.
\newblock On the usefulness of type and liveness accuracy for garbage
  collection and leak detection.
\newblock {\em {ACM} Trans. Program. Lang. Syst.}, 24(6):593--624, 2002.

\bibitem{HofmannJ03}
M.~Hofmann and S.~Jost.
\newblock Static prediction of heap space usage for first-order functional
  programs.
\newblock In {\em Conference Record of {POPL} 2003: The 30th {SIGPLAN-SIGACT}
  Symposium on Principles of Programming Languages, New Orleans, Louisisana,
  USA, January 15-17, 2003}, pages 185--197, 2003.

\bibitem{inoue88analysis}
K.~Inoue, H.~Seki, and H.~Yagi.
\newblock Analysis of functional programs to detect run-time garbage cells.
\newblock {\em {ACM} Trans. Program. Lang. Syst.}, 10(4):555--578, 1988.

\bibitem{jones89compile}
S.~B. Jones and D.~Le~M{\'e}tayer.
\newblock Compile-time garbage collection by sharing analysis.
\newblock In {\em Proceedings of the Fourth International Conference on
  Functional Programming Languages and Computer Architecture}, FPCA '89, pages
  54--74, New York, NY, USA, 1989. ACM.

\bibitem{karkare07liveness}
A.~Karkare, U.~P. Khedker, and A.~Sanyal.
\newblock Liveness of heap data for functional programs.
\newblock In {\em Heap Analysis and Verification, {HAV} 2007, a satellite
  workshop of European Joint Conferences on Theory and Practice of Software,
  {ETAPS} 2007, March 25, 2007, Braga, Portugal}, 2007.

\bibitem{karkare06effectiveness}
A.~Karkare, A.~Sanyal, and U.~P. Khedker.
\newblock Effectiveness of garbage collection in {MIT/GNU} scheme.
\newblock {\em CoRR}, abs/cs/0611093, 2006.

\bibitem{khedker07heap}
U.~P. Khedker, A.~Sanyal, and A.~Karkare.
\newblock Heap reference analysis using access graphs.
\newblock {\em {ACM} Trans. Program. Lang. Syst.}, 30(1), 2007.

\bibitem{lee05static}
O.~Lee, H.~Yang, and K.~Yi.
\newblock Static insertion of safe and effective memory reuse commands into
  ml-like programs.
\newblock {\em Sci. Comput. Program.}, 58(1-2):141--178, 2005.

\bibitem{Mitchell:2009}
N.~Mitchell and C.~Runciman.
\newblock Losing functions without gaining data: another look at
  defunctionalisation.
\newblock In {\em Proceedings of the 2nd {ACM} {SIGPLAN} Symposium on Haskell,
  Haskell 2009, Edinburgh, Scotland, UK, 3 September 2009}, pages 13--24, 2009.

\bibitem{mohri00regular}
M.~Mohri and M.-J. Nederhof.
\newblock Regular approximation of context-free grammars through
  transformation.
\newblock In {\em Robustness in Language and Speech Technology}. Kluwer
  Academic Publishers, 2000.

\bibitem{ONeill}
M.~E. O'Neill and F.~W. Burton.
\newblock Smarter garbage collection with simplifiers.
\newblock In {\em Proceedings of the 2006 workshop on Memory System Performance
  and Correctness, San Jose, California, USA, October 11, 2006}, pages 19--30,
  2006.

\bibitem{Jones87}
S.~L. {Peyton Jones}.
\newblock {\em The Implementation of Functional Programming Languages}.
\newblock Prentice-Hall, 1987.

\bibitem{rojemo96lag}
N.~R{\"{o}}jemo and C.~Runciman.
\newblock Lag, drag, void and use - heap profiling and space-efficient
  compilation revisited.
\newblock In {\em Proceedings of the 1996 {ACM} {SIGPLAN} International
  Conference on Functional Programming {(ICFP} '96), Philadelphia,
  Pennsylvania, May 24-26, 1996.}, pages 34--41, 1996.

\bibitem{shaham01heap}
R.~Shaham, E.~K. Kolodner, and S.~Sagiv.
\newblock Heap profiling for space-efficient {Java}.
\newblock In {\em Proceedings of the 2001 {ACM} {SIGPLAN} Conference on
  Programming Language Design and Implementation (PLDI), Snowbird, Utah, USA,
  June 20-22, 2001}, pages 104--113, 2001.

\bibitem{shaham02estimating}
R.~Shaham, E.~K. Kolodner, and S.~Sagiv.
\newblock Estimating the impact of heap liveness information on space
  consumption in {Java}.
\newblock In {\em Proceedings of The Workshop on Memory Systems Performance
  {(MSP} 2002), June 16, 2002 and The International Symposium on Memory
  Management {(ISMM} 2002), June 20-21, 2002, Berlin, Germany}, pages 171--182,
  2002.

\bibitem{ran.shaham-sas03}
R.~Shaham, E.~Yahav, E.~K. Kolodner, and S.~Sagiv.
\newblock Establishing local temporal heap safety properties with applications
  to compile-time memory management.
\newblock In {\em Static Analysis, 10th International Symposium, {SAS} 2003,
  San Diego, CA, USA, June 11-13, 2003, Proceedings}, pages 483--503, 2003.

\bibitem{shao00efficient}
Z.~Shao and A.~W. Appel.
\newblock Efficient and safe-for-space closure conversion.
\newblock {\em {ACM} Trans. Program. Lang. Syst.}, 22(1):129--161, 2000.

\bibitem{Shivers:1988}
O.~Shivers.
\newblock Control-flow analysis in scheme.
\newblock In {\em Proceedings of the {ACM} SIGPLAN'88 Conference on Programming
  Language Design and Implementation (PLDI), Atlanta, Georgia, USA, June 22-24,
  1988}, pages 164--174, 1988.

\bibitem{Thomas19951}
S.~Thomas.
\newblock Garbage collection in shared-environment closure reducers:
  Space-efficient depth first copying using a tailored approach.
\newblock {\em Information Processing Letters}, 56(1):1 -- 7, 1995.

\bibitem{tofte98region}
M.~Tofte and L.~Birkedal.
\newblock A region inference algorithm.
\newblock {\em {ACM} Trans. Program. Lang. Syst.}, 20(4):724--767, 1998.

\bibitem{wadler88deforest}
P.~Wadler.
\newblock Deforestation: {Transforming} programs to eliminate trees.
\newblock In {\em {ESOP} '88, 2nd European Symposium on Programming, Nancy,
  France, March 21-24, 1988, Proceedings}, pages 344--358, 1988.

\end{thebibliography}
\bibliographystyle{abbrv}
%\clearpage
\appendix
\begin{figure*}[b!]
\section{Complete Minefield Semantics}\label{sec:minefield}

\begin{center}\footnotesize
\renewcommand{\arraystretch}{1.2}
\begin{tabular}{|c|c|c|}
\hline
Premise & Transition & Rule name \\ 
\hline
\hline
&\makecell{ $\rho, (\rho', \ell, e, \cred{\sigma'})\!:\!S,
  \heap, \kappa, \cred{\sigma}$  $\rightsquigarrow \rho', S, \heap[\ell :=
    \kappa], e, \cred{\sigma'}$ }   &  {\sc const}
\\
\hline
& \makecell[t]{$\rho, (\rho', \ell, e, \cred{\sigma'})\!:\!S, \heap,
  (\CONS~x~y), \cred{\sigma}$  $\rightsquigarrow$ \\ $
  \rho', S, \heap[\ell := (\rho(x),\rho(y))], e, \cred{\sigma'}$}     &  {\sc cons} \\
\hline
\makecell[t]{ $\rho(x) \mbox{ is } \bot$} & $\rho, S,
  \heap, (\CAR~x), \cred{\sigma} \rightsquigarrow \bang$   &
{\sc car-bang} 
\\
\hline
\makecell[t]{$\heap(\rho(x)) \mbox{ is } (v, d)$} & \makecell[t]{$\rho, (\rho', \ell, e,
  \cred{\sigma'} )\!:\!S, \heap, (\CAR~x), \cred{\sigma}$  $
  \rightsquigarrow $ \\ $\rho', S, \heap[\ell := v], e, \cred{\sigma'}$}      &
{\sc car-select} \\
\hline
\makecell[t]{ $\heap(\rho(x))
\mbox{ is } \langle s, \rho'\rangle$ }& \makecell[t]{$\rho, S,
  \heap, (\CAR~x), \cred{\sigma} \rightsquigarrow $\\ $  \rho', (\rho, \rho(x),
  (\CAR~x), \cred{\sigma})\!:\!S, \heap, s, \cred{(\clazy \cup \acar)\sigma }$ }        &
{\sc car-clo}\\
\hline
\makecell[t]{$\heap(\rho(x)) \mbox{ is } (\langle s, \rho'\rangle, d)$} & \makecell[t]{$\rho,\, S,\,  \heap,\,
(\CAR~x), \cred{\sigma} \rightsquigarrow$ \\ $ \rho', \,(\rho, addr(\langle
s, \rho'\rangle), (\CAR~x),\cred{\sigma} )\!:\!S,\, \heap,\, s, \, \cred{\sigma}$ }     &
{\sc car-1-clo} \\
%% \hline

%% $\heap(\rho(x)) \mbox{ is } \langle s, \rho'\rangle$ &\makecell{ $\rho, S,
%%   \heap, (\CAR~x), \cred{\sigma}$  $\rightsquigarrow \rho', (\rho, x,
%%   (\CAR~x), \cred{\sigma})\!:\!S, \heap, s, \cred{(\clazy \cup \acar)\sigma }$}      &
%% {\sc car-clo}
%% \\
\hline
\makecell[t]{$\rho(x) \mbox{ is } \bot$  or $\rho(y) \mbox{ is }
  \bot$} & $\rho, S,
  \heap, (+~x~y), \cred{\sigma} \rightsquigarrow \bang$   &
{\sc prim-bang} 
\\
\hline

\makecell[t]{$\heap(\rho(x)), \heap(\rho(y)) \in \mathbb{N}$}
 & 

\makecell[t]{$\rho, (\rho', \ell, e, \cred{\sigma'})\!:\!S, \heap,
  (+~x~y), \cred{\sigma}$  $\rightsquigarrow$ \\ $\rho', S, \heap[\ell
    := \heap(\rho'(x)) + \heap(\rho'(y))], e, \cred{\sigma'}$}     
 &
{\sc prim-add} \\
\hline
\makecell[t]{$\heap(\rho(x)) \mbox{ is } \langle s, \rho'\rangle$} &\makecell[t]{$\rho, S,
  \heap, (+~x~y), \cred{\sigma}$  $\rightsquigarrow$ \\ $ \rho', (\rho, \rho(x),
  (+~x~y), \cred{\sigma})\!:\!S, \heap, s, \cred{\clazy\sigma}$}      &
{\sc prim-1-clo} \\
\hline
\makecell[t]{$\heap(\rho(y)) \mbox{ is } \langle s, \rho'\rangle $} & \makecell[t]{$\rho,
  S, \heap, (+~x~y), \cred{\sigma}$  $\rightsquigarrow$ \\ $ \rho', (\rho,\rho(y),
  (+~x~y), \cred{\sigma})\!:\!S, \heap, s, \cred{\clazy\sigma}$}      &
{\sc prim-2-clo} \\
\hline
\makecell[t]{$\mathit{f}~\mbox{defined as}$
$~(\DEFINE~(f~\myvec{y})~e_{\mathit{f}})$}  & \makecell[t]{$\rho, S, \heap,
  (f~\myvec{x}), \cred{\sigma}$  $\rightsquigarrow$\\$ [\myvec{y} \mapsto
    \rho(\myvec{x})], S, \heap, e_{\mathit{f}}, \cred{\sigma}$}      &
{\sc funcall} \\
\hline
\makecell[t]{$\cred{GC(\rho_1, S_1, \heap_1, (\LET~x\leftarrow
  s~\IN~e), \sigma) = (\rho, S, \heap)}$,\\$\ell$ is a new location}& \makecell[t]{$\rho, S, \heap, (\LET~x\leftarrow
  s~\IN~e), \cred{\sigma}$  $ \rightsquigarrow$ \\ $ \rho\oplus[x
    \mapsto \ell], S, \heap[\ell := \langle s, \lfloor\rho\rfloor_{FV(s)}, \sigma_x\rangle], e, \cred{\sigma}$ \\
    where $\sigma_x\ =\  \lfloor\mathcal{L}(e,\sigma,\Lfonly)\rfloor_{\{x\}}$} &
{\sc let} \\ 
\hline
\makecell[t]{ $\rho(x) \mbox{ is } \bot$} & $\rho, S,
  \heap, (\pi:\SIF~\psi:x~e_1~e_2), \cred{\sigma} \rightsquigarrow \bang$   &
{\sc if-bang} 
\\
\hline
\makecell[t]{$\heap(\rho(x)) \ne 0$}  & \makecell{$\rho, S, \heap, (\pi:\SIF~\psi:x~e_1~e_2),
  \cred{\sigma}$  $\rightsquigarrow \rho, S, \heap,  e_1, \cred{\sigma}$} & {\sc if-true} \\
\hline
\makecell[t]{$\heap(\rho(x)) = 0$} & \makecell{$\rho, S, \heap, (\pi:\SIF~\psi:x~e_1~e_2),
  \cred{\sigma}$   $\rightsquigarrow
\rho, S, \heap,  e_2, \cred{\sigma}$} & {\sc if-false} \\
\hline
\makecell[t]{$\heap(\rho(x)) = \langle s, \rho' \rangle $} & \makecell[t]{$\rho, S, \heap,
  (\pi:\SIF~\psi:x~e_1~e_2), \cred{\sigma}$ $\rightsquigarrow$ \\ $
\rho', (\rho, \rho(x), (\SIF~x~e_1~e_2),  \cred{\sigma})\!:\!S, \heap, s,
\cred{\clazy\sigma}$}
&
{\sc if-clo} \\
\hline
\makecell[t]{
  $\rho(x) \mbox{ is } \bot$} & $\rho, S,
  \heap, (\SRETURN~x), \cred{\sigma} \rightsquigarrow \bang$   &
{\sc return-bang} 
\\
\hline
\makecell[t]{$\heap(\rho(x))~\mbox{is in WHNF with
    value}~v$} & \makecell[t]{$\rho,
  (\rho', \ell, e, \cred{\sigma'})\!:\!S, \heap, (\SRETURN~x), \cred{\sigma}$  $\rightsquigarrow$\\$ \rho', S, \heap[\ell := v], e, 
  \cred{\sigma'}$} &
{\sc return-whnf}\\
\hline
\makecell[t]{$\heap(\rho(x)) = \langle s, \rho' \rangle $} & \makecell[t]{$\rho, S, \heap,
  (\SRETURN~x), \cred{\sigma}$  $
  \rightsquigarrow$ \\
$\rho',~ (\rho, \rho(x), (\SRETURN~x), \cred{\sigma})\!:\!S, \heap,  s,
  \cred{\sigma}$} &
{\sc return-clo} \\
\hline
\end{tabular}

\caption {Minefield semantics. The differences with the small-step
  semantics have been highlighted by shading. \label{fig:minefield-semantics}}
\end{center}

\end{figure*}

\end{document}